\documentclass[english,submission,copyright,creativecommons]{eptcs}
\providecommand{\event}{JAC 2012} 
\usepackage{breakurl}             

\usepackage{amsmath,amsthm,amsfonts,amssymb}
\usepackage{graphicx,url,tikz}
\usetikzlibrary{patterns}
\usetikzlibrary{snakes}
\usetikzlibrary{matrix}
\usepackage{babel}
\usepackage{stmaryrd}
\usepackage{verbatim,bm}

\definecolor{incol}{rgb}{0,0.6,0}
\definecolor{outcol}{rgb}{0.6,0,0}
\definecolor{azul}{RGB}{0,0,200}
\definecolor{rojo}{RGB}{205,0,0}
\definecolor{verde}{rgb}{0.3,0.6,0.3}

\newtheorem{thm}{Theorem}
\newtheorem{cor}[thm]{Corollary}
\newtheorem{prop}[thm]{Proposition}
\newtheorem{lem}[thm]{Lemma}
\newtheorem{rem}[thm]{Remark}
\newtheorem{exa}[thm]{Example}

\theoremstyle{definition}
\newtheorem{defn}{Definition}

\newcommand\definecal[1] 
{
  \expandafter\newcommand
  \csname #1cal\endcsname
  {\mathcal{#1}}
}

\newcommand\Zset{\mathbb{Z}}
\newcommand\Nset{\mathbb{N}}
\newcommand\Qset{\mathbb{Q}}

\newcommand\llb{\llbracket} 
\newcommand\rrb{\rrbracket} 

\title{Linear functional classes over cellular automata }

\author{Ana\"el~Grandjean
 \institute{\'Ecole normale sup\'erieure de Lyon \\
15 parvis Ren\'e Descartes - BP 7000, 69\,342 Lyon Cedex 07, FRANCE}\\
\institute{GREYC, Universit\'{e} de Caen Basse-Normandie, CNRS, \\
Campus C\^{o}te de Nacre, Boulevard du Mar\'{e}chal Juin, 14\,000 Caen, France}
\email{anael.grandjean@ens-lyon.fr}
\and
Ga\'etan~Richard \qquad \qquad V\' eronique~Terrier
\institute{GREYC,  Universit\'{e} de Caen Basse-Normandie, CNRS, \\
Campus C\^{o}te de Nacre, Boulevard du Mar\'{e}chal Juin, 14\,000 Caen, France}
\email{\qquad gaetan.richard@unicaen.fr \qquad \qquad veronique.terrier@unicaen.fr}
}

\def\titlerunning{Linear functional classes over cellular automata}
\def\authorrunning{A.Granjean , G.~Richard \& V.~Terrier}

\begin{document}
\maketitle

\begin{abstract}
  Cellular automata are a discrete dynamical system which models
  massively parallel computation.  Much attention is devoted to
  computations with small time complexity for which the parallelism
  may provide further possibilities. In this paper, we investigate the
  ability of cellular automata related to functional computation.  We
  introduce several functional classes of low time complexity which
  contain ``natural'' problems.  We examine their inclusion
  relationships and emphasize that several questions arising from this
  functional framework are related to current ones coming from the
  recognition context.  We also provide a negative result which
  explicits limits on the information transmission whose consequences
  go beyond the functional point of view.
\end{abstract}

\section*{Introduction}

Introduced in the forties to model
self-replication~\cite{VonNeumann:1966}, cellular automata are a
discrete dynamical model composed of an infinite line of cells endowed
with a state chosen among a finite alphabet. Dynamic is achieved
in discrete time by applying uniformly and synchronously a local rule
to each cell.

From the dynamical point of view, this system has been widely used to
model phenomena issued from different fields of research. It is often
cited as a representative of complex systems --- systems that can
exhibit a complex behavior even starting from simple rules. Thanks to
its simple formal definition, many results have also been achieved
over its dynamics (see~\cite{Kari:2005}).

This model has also been studied \emph{per se} as a theoretical model
of massively parallel computation. For this purpose, one usually gives
as input a \emph{finite} word over the line of cells and waits until
some predefined state occurs in the evolution to determine whether
this word is accepted or not. The number of steps needed is considered
as the \emph{time}.  Among all possible complexity classes, the mainly
studied and interesting ones are real time --- the minimal time needed
to take into account the whole information ---
and linear time~\cite{Smith:1971}. 
Central questions concerning their computational power and their
limits remain unanswered.

To tackle such problems, a possible way
is to extend the study
of the model to the functional point of view. 
A first significant step was made
by M.~Kutrib and A.~Malcher who have investigated iterative arrays (a variant of cellular automata) as transducers~\cite{Kutrib:2010}
and reported 
several interesting results. 
For the device they considered, both input and output modes are sequential. 
In this
paper, we study the case where input and output are fed and retrieved in parallel and examine the corresponding small complexities classes.


After specifying different possible definitions of functional classes
in Section~\ref{sec:def}, we give several meaningful examples in
Section~\ref{sec:prop} along with some generic framework to build such
algorithms.  We present some closure properties, a linear acceleration
algorithm, the basic relationships between classes and some links with
classical questions on CA recognition ability in
Section~\ref{sec:cons}. In Section~\ref{sec:hier}, we prove separation
results over the classes by providing one specific impossibility of
behavior for cellular automata.  The latter result
(Theorem~\ref{thm:impo}) is of interest in itself and opens new
perspectives to achieve negative results over cellular automata.

\section{Definitions}\label{sec:def}


Basically, a cellular automata (CA) is a one-dimensional array of finite automata (the cells) indexed by~$\Zset$. The cells range over a finite set $S$, \emph{the set of states}, and evolve synchronously in discrete time steps. At each step, the state of each cell changes according to its own state and the states of its nearest neighbors. All cells have the same \emph{local transition rule} $f$.
Formally denoting $(c,t)$ the cell $c$ at time $t$ and  $\langle c,t\rangle$ its state, we have: 
$\langle c,t+1\rangle=f(\langle c-1,t\rangle,\langle c,t\rangle,\langle c+1,t\rangle )$.

A \emph{configuration} is the sequence of cell states at a given time.
To represent the evolution of a cellular automaton starting from a
given configuration, one convenient representation is a
\emph{space-time diagram} which consists in piling-up the
configurations at successive time steps.

Viewed as a computational model, CA operate on finite words.  Although
different alternatives may be relevant, we make the following choices
for the rest of the paper.  First, we only consider parallel input
mode. That means the input sequence $w$ is supplied at initial time to
the array: $\langle i,0\rangle = w_i$ for $0 \leq i < |w|$.  Second,
we will assume that the computation is linearly bounded in space: only
a fixed number of cells, equal to the length of the input, are active.
In practice, when the input length is $n$, the cells not in range
$\{0,\cdots,n-1\}$ will remain in a persistent state $\#$ during all
the computation.  Actually, this bound coincides with the space
consumed by small time computations, i.e., those computations which
attract our attention.

We also need to specify the output mode.  Obviously, it depends on
whether we are looking at recognition or functional computation.  With
its output of yes/no type, the recognition case is the simplest one.
Before examining how the output could be retrieved in the functional
case, we first recall the definitions related to recognition.

To use cellular automata as a recognizer, two subsets of the states
set $S$ are specified: the set of accepting states and the set of
rejecting states.  The cell indexed by $0$ is chosen to be the output
cell which determines the acceptance.  So an input word $w \in S^*$ is
said to be accepted (resp. rejected) in time $t\in\Nset$ if the
cell~$0$ enters an accepting state (resp. a rejecting state) at
time~$t$; and for all time less than $t$, the output cell is neither
in an accepting nor in a rejecting state. The language
\emph{recognized} by the automaton is the set of words it eventually
accepts.  A CA works within time $\tau : \Nset\to\Nset$ if every
word~$w$ is accepted or rejected in time $\tau(|w|)$.  Among the time
complexities, the small ones --- real time and linear time --- are of
major of interest.

\begin{defn}
  A language $L \subset S^*$ is recognized in \emph{real time}
  (resp. \emph{linear time}) if it corresponds to the set of words $w$
  recognized by a cellular automaton in time $|w|$ (resp. $k |w|$ for
  some $k \in \Qset$ with $k>1$).
\end{defn}

 
Since they have been introduced by A.~R.~Smith \cite{Smith:1971}, a
lot of work has been done to study these complexity classes in order
to better understand parallel feature. Despite a number of interesting
results (see \cite{Terrier:2011} for a survey), the basic question
whether real time and linear time classes differ or not, is still
open.

One can also take as alternative definition the case where the output
cell is located in the middle of the input word.  The minimal time for
the output cell to know the whole input, is reduced to half of the
input length.  In fact, this notion corresponds to real time CA
restricted to one-way communication which is known to be strictly less
powerful than real time CA with two-way communication.\\ 




Let us come now to the topic of our paper: the functional issue.
When trying to formally define functional variant, one problem arises. For
the real time complexity, the minimal time to obtain the whole information on the
word differs according to the position of the cell inside the
word. This gives birth to two variations over real functional time
according whether we require the output to be synchronized. The
resulting classes depicted in Figure~\ref{fig:frt} are defined as
follows:

\begin{defn}
  A function $\phi: I^* \to O^*$ is called \emph{computable in strict
    real time} if there exists a cellular automaton $(S,f)$ and a
  projection $\pi : S \to O$ such that, on any input word $w\in I^*$,
  we have for any $i \in \llb 0,\frac{|w|-1}{2} \rrb$, $\pi(\langle
  i,|w|-i \rangle) = \phi(w)_i$ and for any $i \in \llb
  \frac{|w|-1}{2},|w|-1 \rrb$, $\pi(\langle i,i+1 \rangle) =
  \phi(w)_i$.
\end{defn}

\begin{defn}
  A function $\phi: I^* \to O^*$ is said \emph{computable in
    synchronous real time} if there exists a cellular automaton
  $(S,f)$ and a projection $\pi : S \to O$ such that, on any input
  word $w\in I^*$, we have $\pi(\langle i,|w| \rangle) = \phi(w)_i$
  for any $i \in \llb 0,|w|-1 \rrb$ .
\end{defn}

\begin{figure}[!htp]
  \centering
  \begin{tabular}{c     c     c}
    \begin{tikzpicture}[scale=.7]
\draw (-1,0) rectangle (0,1);\draw (-0.5,0.5) node {$\sharp$};
\draw (3,0) rectangle (4,1);\draw (3.5,0.5) node {$\sharp$};
\draw (0,0) rectangle (3,1);\draw (1.5,0.5) node {$\omega$};
\draw (0,1) -- (1.5,2.5);\draw (3,1) -- (1.5,2.5);
\draw [<->][very thick] (-0.5,1) -- node [left = 1pt]{$| \omega |$} (-0.5,4);
\draw [blue,thick] (0,4) -- (1.5,2.5) -- (3,4);
\draw[blue,thick] (0,5) -- (1.5,3.5) -- (3,5);
\draw (1.5,3.2) node{$f(\omega)$};
\draw (-1,4) rectangle (0,5);\draw (-0.5,4.5) node {$\sharp$};
\draw (3,4) rectangle (4,5);\draw (3.5,4.5) node {$\sharp$};
    \end{tikzpicture} &
   \begin{tikzpicture}[scale=.7]
\draw (-1,0) rectangle (0,1);\draw (-0.5,0.5) node {$\sharp$};
\draw (3,0) rectangle (4,1);\draw (3.5,0.5) node {$\sharp$};
\draw (0,0) rectangle (3,1);\draw (1.5,0.5) node {$\omega$};
\draw[blue,thick] (0,4) rectangle (3,5);\draw (1.5,4.5) node {$f(\omega )$};
\draw (0,1) -- (3,4);\draw (3,1) -- (0,4);
\draw [<->][very thick] (-0.5,1) -- node [left = 1pt]{$| \omega |$} (-0.5,4);
\draw (-1,4) rectangle (0,5);\draw (-0.5,4.5) node {$\sharp$};
\draw (3,4) rectangle (4,5);\draw (3.5,4.5) node {$\sharp$};
    \end{tikzpicture} &
   \begin{tikzpicture}[scale=.7]
\draw (-1,0) rectangle (0,1);\draw (-0.5,0.5) node {$\sharp$};
\draw (3,0) rectangle (4,1);\draw (3.5,0.5) node {$\sharp$};
\draw (0,0) rectangle (3,1);\draw (1.5,0.5) node {$\omega$};
\draw[blue,thick] (0,5) rectangle (3,6);\draw (1.5,5.5) node {$f(\omega )$};
\draw (0,1) -- (3,4);\draw (3,1) -- (0,4);
\draw [<->][very thick] (-0.5,1) -- node [left = 1pt]{$k |\omega|$} (-0.5,5);
\draw (-1,5) rectangle (0,6);\draw (-0.5,5.5) node {$\sharp$};
\draw (3,5) rectangle (4,6);\draw (3.5,5.5) node {$\sharp$};
    \end{tikzpicture} \\
    (a) Strict real time &
    (b) Synchronous real time &
    (c) Linear time \\
  \end{tabular}
 \caption{Functional classes over cellular automata}
  \label{fig:frt}
\end{figure}
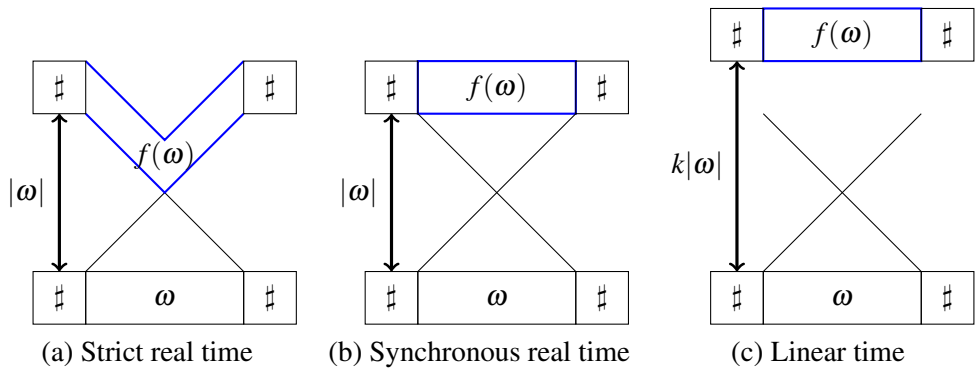

We will also consider linear time complexity (in this case, the two
output variants are equivalent as proved in
proposition~\ref{lem:lin_equiv}).

\begin{defn}
  A function $\phi: I^* \to O^*$ is said \emph{computable in linear
    time} if there exists $k \in \Qset$ with $k>1$, a cellular
  automaton $(S,f)$ and a projection $\pi : S \to O$ such that, on any
  input word $w\in I^*$, $\pi(\langle i,k|w| \rangle) = \phi(w)_i$ for
  any $i \in \llb 0,|w|-1 \rrb$ .
\end{defn}


In the above definitions, we refer to sets of sites on which the
outputs are displayed.  In each case, the sequence of sites may be
distinguished either, for strict real time, in making use of two
signals initiated from each input extremity or, for synchronous time,
by the way of \emph{Firing Squad Synchronization} solutions
(see~\cite{Culik:1989a,Mazoyer:1996}). We also observe that, for both
the real time complexities, it is possible to answer one step sooner
by anticipating what would happen when the end of the input is
reached. But the drawback is that we lose the capacity to explicit the
set of output sites by marking them.

In addition, one can note that linear time is not affected by changing
the input or output mode from parallel to sequential whereas this is
not the case for (both) real time.


\section{Examples}\label{sec:prop}

To start the study, we shall give several meaningful examples
illustrating the range of functional classes and present some
interesting (generic) algorithms.

One first easy remark is that functional classes are a generalization
of detection.

\begin{lem}\label{lemRec2Fnl}
  If $L \subset I^*$ is a language recognizable in real (resp. linear)
  time then the function $f: I^* \to \{1,0\}^*$ which is defined by
  $f(w) = 10^{|w| - 1}$ if $w \in L$ and $f(w) = 0 0^{|w|-1}$
  otherwise is computable in strict real time (resp. linear time).
\end{lem}

\begin{proof}
  It is sufficient to take the recognition automaton and  send all output sites to $0$
  but the first one.
\end{proof}

Let us first look at some simple examples which still use the power of
parallelism.

\begin{exa}
  The mirror is computable in synchronous real time.
\end{exa}

\begin{proof}
  Two possibilities are depicted in Figure~\ref{fig:algo-mirror}.  The
  left algorithm has been exhibited by M.~Kutrib
  in~\cite{Kutrib:2008}. The right one is a variant with symmetric
  features.
  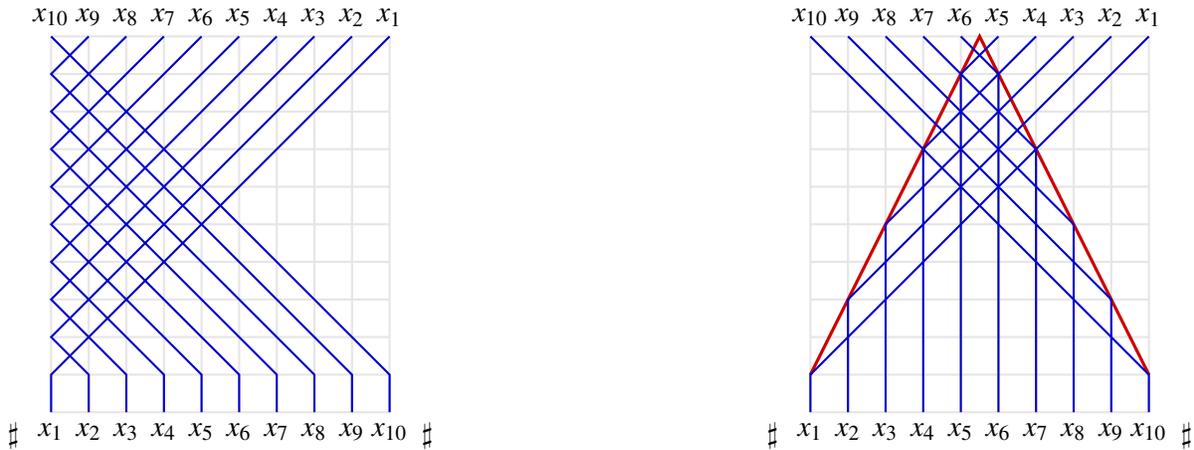
\begin{figure}[!htp]
    \centering
          \begin{tikzpicture}[thick,scale=.5]
\def\nb{10}
\def\mil{5}
\node [below] (bg) at (0,0)  {$\sharp$};
\node [below] (bd) at (\nb+1,0)  {$\sharp$};
\draw [black!10] (1,0) grid (\nb,\nb);

\foreach \c in {1,...,\nb} {
\coordinate [label=below:$x_{\c}$] (x\c) at (\c,0);
\coordinate [label=above:$x_{\c}$] (y\c) at (\nb+1-\c,\nb);
\coordinate (i\c) at (1,\c);
}

\foreach \c in {1,...,\nb}
\draw [azul] (x\c)--++(0,1)--(i\c)--(y\c);
\end{tikzpicture}
          \hfill
          \begin{tikzpicture}[thick,scale=.5]
\def\nb{10}
\def\mil{5}
\node [below] (bg) at (0,0)  {$\sharp$};
\node [below] (bd) at (\nb+1,0)  {$\sharp$};
\coordinate(bg1) at (1,1);
\coordinate (bd1) at (\nb,1);
\coordinate (bm) at (\nb/2+.5,\nb);

\draw [black!10] (1,0) grid (\nb,\nb);

\foreach \c in {1,...,\nb} {
\coordinate[label=below:$x_{\c}$] (x\c) at (\c,0) ;
\coordinate [label=above:$x_{\c}$] (y\c) at (\nb+1-\c,\nb) ;
}
\foreach \c in {\mil,...,\nb} \coordinate (i\c) at (\c,2*\nb-2*\c+1);
\foreach \c in {1,...,\mil} \coordinate (i\c) at (\c,2*\c-1) ;

\draw [rojo,very thick] (bg1)--(bm)--(bd1);

\foreach \c in {1,...,\nb}
\draw [azul] (x\c)--(i\c)--(y\c);
\end{tikzpicture}
    \caption{Computing the mirror in synchronous real time}
    \label{fig:algo-mirror}
  \end{figure}
\end{proof}

\begin{exa}\label{ex:fg}
  The functions 
  $f(u_1\cdots u_n)= u_{n/2+1}\cdots u_nu_1\cdots u_{n/2}$ and
  $g(u_1\cdots u_n)= u_{n/2}\cdots u_1u_n\cdots u_{n/2}$
  are computable in strict real-time.
\end{exa}
\begin{proof}
  The flow of data through the space time-diagram is depicted for the
  two functions $f$ and $g$ in Figure~\ref{fig:uv2vu-uv2urvr}.
  For the construction of $g$, the space can be reduced to the bounded space of $n$ cells in folding the figure.\\
  Note that the composition of these two functions is the mirror
  function. 
\end{proof}
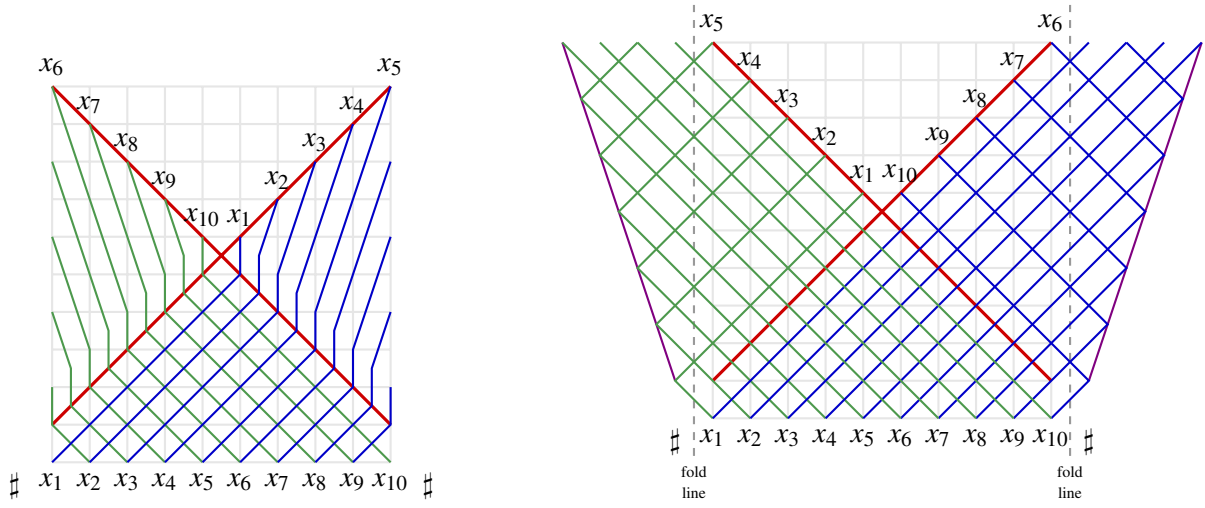
\begin{figure}[!htp]
    \centering
    \begin{tikzpicture}[thick,scale=.5]
\def\nb{10} 
\def\mil{5}
\def\mild{6}
\def\nbg{9}

\node [below] (bg) at (0,0)  {$\sharp$};
\node [below] (bd) at (\nb+1,0)  {$\sharp$};

\coordinate(bg1) at (1,1);
\coordinate (bd1) at (\nb,1);
\coordinate(hg1) at (1,\nb);
\coordinate (hd1) at (\nb,\nb);

\draw [black!10] (1,0) grid (\nb,\nb);
\draw [rojo,very thick] (bg1)--(hd1) (bd1)--(hg1);

\foreach \c in {1,...,\nb} {
\coordinate[label=below:$x_{\c}$] (x\c) at (\c,0);
\coordinate (d\c) at (\c*.5,\c*.5) ;
\coordinate (e\c) at (\mil+\c*.5+.5,\mil-\c*.5+.5) ;
}

\foreach \c in {\mild,...,\nb} {
\coordinate [label=above:$x_{\c}$] (y\c) at (\c-\mil,\nb+\mil-\c+1) ;
\coordinate (z\c) at (\nb,2*\nb-2*\c) ;
}
\foreach \c in {1,...,\mil} {
\coordinate  (y\c) at (1,2*\c-2) ;
\coordinate [label=above:$x_{\c}$] (z\c) at (\c+\mil,\c+\mil) ;
}

\foreach \c in {2,...,\nb} \draw [verde] (x\c)--(d\c)--++(0,1)--(y\c);
\foreach \c in {1,...,\nbg} \draw [azul] (x\c)--(e\c)--++(0,1)--(z\c);
\end{tikzpicture}
    \hfill\begin{tikzpicture}[thick,scale=.5,text centered]
\def\nb{10} 
\def\mil{5}
\def\mild{6}
\def\nbg{9}

\node [below] (bg) at (0,0)  {$\sharp$};
\node [below] (bd) at (\nb+1,0)  {$\sharp$};
\coordinate(bg1) at (1,1);
\coordinate (bd1) at (\nb,1);
\coordinate(hg1) at (1,\nb);
\coordinate (hd1) at (\nb,\nb);

\coordinate (p0) at (0,1);
\coordinate (p1) at (-\nb/3+1/3,\nb);
\coordinate (q0) at (\nb+1,1);
\coordinate (q1) at (\nb+\nb/3+2/3,\nb);

\draw [black!10] (1,0) grid (\nb,\nb);
\draw [rojo,very thick] (bg1)--(hd1) (bd1)--(hg1);
\draw[violet]  (p0) -- (p1) (q0) -- (q1);

\foreach \c in {1,...,\nb} {
\coordinate[label=below:$x_{\c}$] (x\c) at (\c,0);
}

\foreach \c in {\mild,...,\nb} {
\coordinate  (y\c) at (-\c*.5+.5,\c*1.5-.5) ;
\coordinate [label=above:$x_{\c}$](z\c) at (\nb+\mild-\c,\nb+\mild-\c) ;
\coordinate (e\c) at (\nb*1.5+1-\c*.5,\nb*1.5-\c*1.5+1) ;
}
\foreach \c in {1,...,\mil} {
\coordinate  [label=above:$x_{\c}$](y\c) at (\mild-\c,\mil+\c) ;
\coordinate  (z\c) at (\c+\mil,\c+\mil) ;
\coordinate (d\c) at (-\c*.5+.5,\c*1.5-.5) ;
}

\draw[dashed,gray!80] (.5,-1) -- (.5,\nb+1) (\nb+.5,-1) -- (\nb+.5,\nb+1);
\node[below,text width=6mm] at (.5,-1) {\tiny fold\\[-2mm]line};
\node[below,text width=6mm] at (\nb+.5,-1) {\tiny fold\\[-2mm]line};

\foreach \c in {1,...,\mil} \draw [verde] (x\c)--(d\c)--(y\c);
\foreach \c in {\mild,...,\nb} \draw [azul] (x\c)--(e\c)--(z\c);

\draw [azul] (x5)--(13.5,8.5)--(12,10);
\foreach \c in {1,...,4} \draw [azul] (x\c)--(\nb+\c,\nb);
\draw [verde] (x6)--(-2.5,8.5)--(-1,10);
\foreach \c in {7,...,10}\draw [verde]  (x\c)--(\c-\nb,\nb);
\end{tikzpicture}
    \caption{Computing $f$ and $g$ in strict real time}
    \label{fig:uv2vu-uv2urvr}
  \end{figure}

\begin{exa}
The function 
  $h(\flat^nu_1\cdots u_n)= u_n\cdots u_1\flat^n$
  is computable in strict real-time.
\end{exa}

\begin{proof}
  See Figure~\ref{fig:algomiddle}. Knowing the middle of the word at time $1$ allows to set up a signal that serves as a reflector.  
\end{proof}

\begin{figure}[!htp]
  \centering
  \begin{tikzpicture}[thick,scale=.5,text centered]
\def\nb{10} 
\def\mil{5}
\def\mild{6}
\def\nbg{9}

\draw [help lines] (1,0) grid (\nb,\nb);
\node [below] (bg) at (0,0)  {$\sharp$};
\node [below] (bd) at (\nb+1,0)  {$\sharp$};

\draw [rojo,very thick] (1,1) -- (\nb,\nb) (\nb,1) -- (1,\nb);
\draw[line width=1mm,violet]  (\mil,1) -- (\mil/2,1+\mil/2) -- (.5,\nb-.5);

\foreach \c in {1,...,\mil} \coordinate[label=below:$\flat$] (x\c) at (\c,0);

\foreach \c in {\mild,...,\nb} {
\coordinate[label=below:$x_{\c}$] (x\c) at (\c,0);
\coordinate  (y\c) at (\mil-\c*.5+.5,-\mil+3*\c/2-.5) ;
\coordinate [label=above:$x_{\c}$](z\c) at (1+\nb-\c,\c);
}

\foreach \c in {\mild,...,\nb}\draw [verde]  (x\c) -- (y\c) -- (z\c);
\end{tikzpicture}
  \caption{Computing $h$ in strict real time}
  \label{fig:algomiddle}
\end{figure}
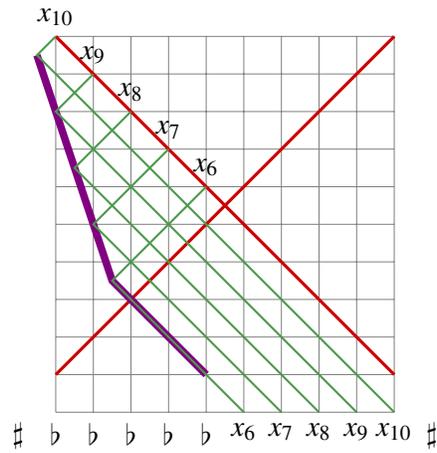

In a recent result, T.~Worsch and H.~Nishio have proved that sorting
binary numbers of the same size is computable in synchronous real time
\cite{Worsch:2010}.  The algorithm is based on an odd-even sort and
uses some clever adaptations to achieve synchronous real time.  Here,
we shall present a new algorithm to sort in linear time.  Its interest
lies in the ``generic'' method used which applies to several different
problems. The basic idea is to build a \emph{assembly line}: 
the input will be traveling along some path where
\emph{agents} will act on it.  We shall give three significant
examples taking advantage of this method: sorting a sequence,
reordering a cycle in a graph, and marking the connected components of
an undirected graph.

\begin{lem}
  Sorting a sequence can be done in linear time.
\end{lem}

\begin{proof}[Algorithm]

  To do this, we use two layers of states. The lower layer (in black)
  will serve to transmit information whereas the upper layer (in blue)
  will stay still and serve as an agent.

  The basic scheme is the following: the lower level travels to the
  left. If it is greater than the agent (in the upper level), they
  swap place. If the upper level is empty, the lower information
  becomes the agent (see Figure~\ref{fig:algo-sort}). 

  Once the end of the word reaches an agent, it indicates the end of
  computation for this agent. The algorithm is ended by shifting the
  result (in green) back in place (not depicted in the figure).

 \begin{figure}[!htp]
    \centering
    \begin{tikzpicture}[scale=.5]
      \begin{footnotesize}
        
\draw[gray] (0,0) grid +(6,9);
\draw[gray] (-1,3) grid (0,9);
\draw[gray] (-2,5) grid (-1,9);
\draw[gray] (-3,7) grid (-2,9);

\draw (1,0) -- +(0,9);
\draw (5,0) -- +(0,9);


\node at (.66,.3) {$\#$};
\node at (1.66,.3) {$7$};
\node at (2.66,.3) {$4$};
\node at (3.66,.3) {$8$};
\node at (4.66,.3) {$1$};
\node at (5.66,.3) {$\#$};

\node at (.66,1.3) {$7$};
\node at (1.66,1.3) {$4$};
\node at (2.66,1.3) {$8$};
\node at (3.66,1.3) {$1$};
\node at (4.66,1.3) {$\#$};

\node[blue] at (.33,2.7) {$7$};
\node at (.66,2.3) {$4$};
\node at (1.66,2.3) {$8$};
\node at (2.66,2.3) {$1$};
\node at (3.66,2.3) {$\#$};

\node[blue] at (.33,3.7) {$7$};
\node at (-.33,3.3) {$4$};
\node at (.66,3.3) {$8$};
\node at (1.66,3.3) {$1$};
\node at (2.66,3.3) {$\#$};

\node[blue] at (-.66,4.7) {$4$};
\node[blue] at (.33,4.7) {$8$};
\node at (-.33,4.3) {$7$};
\node at (.66,4.3) {$1$};
\node at (1.66,4.3) {$\#$};

\node[blue] at (-.66,5.7) {$7$};
\node[green!50!black] at (.33,5.7) {$8$};
\node at (-1.33,5.3) {$4$};
\node at (-.33,5.3) {$1$};
\node at (.66,5.3) {$\#$};

\node[blue] at (-1.66,6.7) {$4$};
\node[green!50!black] at (-.66,6.7) {$7$};
\node[green!50!black] at (.5,6.5) {$8$};
\node at (-1.33,6.3) {$1$};
\node at (-.33,6.3) {$\#$};

\node[green!50!black] at (-1.66,7.7) {$4$};
\node[green!50!black] at (-.5,7.5) {$7$};
\node[green!50!black] at (.5,7.5) {$8$};
\node at (-2.33,7.3) {$1$};
\node at (-1.33,7.3) {$\#$};

\node[green!50!black] at (-1.5,8.5) {$4$};
\node[green!50!black] at (-.5,8.5) {$7$};
\node[green!50!black] at (.5,8.5) {$8$};
\node[green!50!black] at (-2.66,8.7) {$1$};
\node at (-2.33,8.3) {$\#$};
      \end{footnotesize}
    \end{tikzpicture}
    \caption{Scheme of sorting algorithm}
    \label{fig:algo-sort}
  \end{figure}

  Here, the behavior is depicted using integers to underline the
  scheme. Practically, those numbers are supposed to be encoded in
  binary with a fixed number of cells and require thus some fixed
  size. Every elementary transition of the scheme can be done in
  linear time with respect to the size of integers. Since the number of
  steps of the scheme is proportional to the number of integers the
  resulting algorithm is linear in the size of the input.

\end{proof}

The previous algorithm is not a surprising result and can be probably
presented in a different way. However, the generic idea of the method
can be adapted to the problem of \emph{edge reordering of path}: given
a sequence of edges that form a path in a random order, we want to
reconstruct the order of nodes in the path. For example, the input $
(6,12) (2,6) (1,11) (12,7) (8,1) (7,8)$ should  output
$2,6,12,7,8,1,11$. Intuitively, this problem can be seen as a
sorting problem in which the order is given by local
constraints. Using similar method as previously, the problem can be
solved in linear time.

\begin{lem}
  Edge reordering of a path can be done in linear time.
\end{lem}

\begin{proof}[Algorithm]

  The algorithm also uses two layered states as
  previously. 

  In the first step, elements of both layer are triplets representing
  the beginning, the end and the length of one path (for example,
  $(3,5)_2$ represent the path of length $2$, going from edge $3$ to
  edge $5$). Initially, all edges correspond to path of length
  $1$. The basic operation of an agent consists in merging once the two
  paths present in the cell and storing the relative position of the
  gluing element with respect to the start of the newly created
  path. Once all information are gone through, the result consists of
  a unique element designating the path.
 
  In a second step, the first vertex of the path is send backwards.
  Each agent waits for its reference and after seeing it, can put the
  second vertex at the correct position using the length as a counter.

  The correctness of the algorithm lies in the following properties:
  during the first step, any agent is removing reference to exactly
  one vertex (and all references to this vertex since it is guaranteed
  to only appear twice in the input). The resulting flow transmitted
  to its left neighbour is a valid data. Moreover, if the flow
  returning from the left neighbour is correct, then the missing
  vertex is added at the correct position. Since the last agent exists
  and does its job correctly, a recurrence can prove the correctness
  of the algorithm.

  \begin{figure}[!htp]
    \centering
    \begin{footnotesize}
      \begin{tikzpicture}[xscale=.85,yscale=.55]
          
\draw[gray] (0,0) grid +(6,15);
\draw[gray] (-1,3) grid (0,9);!
\draw[gray] (-2,5) grid (-1,9);
\draw[gray] (-3,7) grid (-2,9);

\draw (1,0) -- +(0,15);
\draw (5,0) -- +(0,15);


\node at (.5,.25) {$\#$};
\node at (1.5,.25) {$(1,6)_1$};
\node at (2.5,.25) {$(4,2)_1$};
\node at (3.5,.25) {$(6,3)_1$};
\node at (4.5,.25) {$(3,4)_1$};
\node at (5.5,.25) {$\#$};

\node at (.5,1.25) {$(1,6)_1$};
\node at (1.5,1.25) {$(4,2)_1$};
\node at (2.5,1.25) {$(6,3)_1$};
\node at (3.5,1.25) {$(3,4)_1$};
\node at (4.5,1.25) {$\#$};

\node[blue!60!black] at (.5,2.75) {$(1,6)_1$};
\node at (.5,2.25) {$(4,2)_1$};
\node at (1.5,2.25) {$(6,3)_1$};
\node at (2.5,2.25) {$(3,4)_1$};
\node at (3.5,2.25) {$\#$};

\node[blue!60!black] at (.5,3.75) {$(1,6)_1$};
\node at (-.5,3.25) {$(4,2)_1$};
\node at (.5,3.25) {$(6,3)_1$};
\node at (1.5,3.25) {$(3,4)_1$};
\node at (2.5,3.25) {$\#$};

\node[blue] at (.5,4.75) {$(1,6)_1$};
\node[blue!60!black] at (-0.5,4.75) {$(4,2)_1$};
\node at (-.5,4.25) {$(1,3)_2$};
\node at (.5,4.25) {$(3,4)_1$};
\node at (1.5,4.25) {$\#$};

\node[blue] at (.5,5.75) {$(1,6)_1$};
\node[blue!60!black] at (-0.5,5.75) {$(4,2)_1$};
\node at (-1.5,5.25) {$(1,3)_2$};
\node at (-.5,5.25) {$(3,4)_1$};
\node at (.5,5.25) {$\#$};

\node[blue] at (.5,6.75) {$(1,6)_1$};
\node[blue] at (-0.5,6.75) {$(3,4)_1$};
\node[blue!60!black] at (-1.5,6.75) {$(1,3)_2$};
\node at (-1.5,6.25) {$(3,2)_2$};
\node at (-.5,6.25) {$\#$};

\node[blue] at (.5,7.75) {$(1,6)_1$};
\node[blue] at (-0.5,7.75) {$(3,4)_1$};
\node[blue] at (-1.5,7.75) {$(1,3)_2$};
\node at (-2.5,7.25) {$(1,2)_4$};
\node at (-1.5,7.25) {$\#$};

\node[blue] at (.5,8.75) {$(1,6)_1$};
\node[blue] at (-0.5,8.75) {$(4,3)_1$};
\node[blue] at (-1.5,8.75) {$(1,3)_2$};
\node[blue!60!black] at (-2.5,8.75) {$(1,2)_4$};
\node at (-2.5,8.25) {$\#$};
        \begin{scope}[shift={(-3,9)}]
          
\draw[gray] (0,0) grid +(6,6);



\node[blue] at (3.5,.75) {$(1,6)_1$};
\node[blue] at (2.5,.75) {$(3,4)_1$};
\node[blue] at (1.5,.75) {$(1,3)_2$};
\node[blue] at (.5,.75) {$(1,2)_4$};
\node[green!50!black] at (.5,.25) {$1$};

\node[blue] at (3.5,1.75) {$(1,6)_1$};
\node[blue] at (2.5,1.75) {$(3,4)_1$};
\node[blue] at (1.5,1.75) {$(1,3)_2$};
\node[blue] at (.5,1.75) {$(2)_3$};
\node[green!50!black] at (1.5,1.25) {$1$};

\node[blue] at (3.5,2.75) {$(1,6)_1$};
\node[blue] at (2.5,2.75) {$(3,4)_1$};
\node[blue] at (1.5,2.75) {$(3)_1$};
\node[blue] at (.5,2.75) {$(2)_2$};
\node[green!50!black] at (2.5,2.25) {$1$};

\node[blue] at (3.5,3.75) {$(1,6)_1$};
\node[blue] at (2.5,3.75) {$(3,4)_1$};
\node[green!50!black] at (1.5,3.25) {$3$};
\node[blue] at (.5,3.75) {$(2)_1$};
\node[green!50!black] at (3.5,3.25) {$1$};

\node[green!50!black] at (3.5,4.25) {$6$};
\node[blue] at (2.5,4.75) {$(3,4)_1$};
\node[green!50!black] at (2.5,4.25) {$3$};
\node[green!50!black] at (.5,4.25) {$2$};
\node[green!50!black] at (4.5,4.25) {$1$};

\node[green!50!black] at (4.5,5.25) {$6$};
\node[green!50!black] at (2.5,5.25) {$4$};
\node[green!50!black] at (3.5,5.25) {$3$};
\node[green!50!black] at (1.5,5.25) {$2$};
\node[green!50!black] at (5.5,5.25) {$1$};
        \end{scope}
      \end{tikzpicture}
 
    \end{footnotesize}
    \caption{Edge reordering along path}
    \label{fig:algo-edge}
  \end{figure}
  
\end{proof} 

At last, we want to present a third use of the generic method. In this
last case, the problem is: given a sequence of edges of an undirected
graph, can we output the same sequence where every edge is marked by a
unique identifier per connected component. 

\begin{rem}
  Connected components marking in undirected graphs is computable in
  linear time.
\end{rem}

\begin{proof}[Algorithm]

  This algorithm (depicted in Figure~\ref{fig:algo-mark}) is a
  variation of the previous one and associates to each vertex the
  label of the smallest vertex in its connected component.  In a first
  step, the agent (given a edge) only replaces any occurrence of the
  greater vertex by the smaller one. During the second step, when
  seeing the result for the smallest one, it duplicate it for the
  greater vertex. The proof of correctness work as in the previous case,
  each agent ``suppress'' one vertex and send valid data to its left
  neighbour in the first step. In the second step, it ``add'' this
  vertex correctly.

    \begin{figure}[!htp]
      \centering
      \begin{tikzpicture}[xscale=.8,yscale=.55]
        \begin{footnotesize}
          
\draw[gray] (0,0) grid +(6,16);
\draw[gray] (-1,3) grid (0,13);
\draw[gray] (-2,5) grid (-1,11);
\draw[gray] (-3,7) grid (-2,9);

\draw (1,0) -- +(0,16);
\draw (5,0) -- +(0,16);


\node at (.5,.25) {$\#$};
\node at (1.5,.25) {$(5,7)$};
\node at (2.5,.25) {$(4,5)$};
\node at (3.5,.25) {$(2,6)$};
\node at (4.5,.25) {$(4,1)$};
\node at (5.5,.25) {$\#$};

\node at (.5,1.25) {$(5,7)$};
\node at (1.5,1.25) {$(4,5)$};
\node at (2.5,1.25) {$(2,6)$};
\node at (3.5,1.25) {$(4,1)$};
\node at (4.5,1.25) {$\#$};

\node[blue] at (.5,2.75) {$(5,7)$};
\node at (.5,2.25) {$(4,5)$};
\node at (1.5,2.25) {$(2,6)$};
\node at (2.5,2.25) {$(4,1)$};
\node at (3.5,2.25) {$\#$};

\node[blue] at (.5,3.75) {$(4,7)$};
\node at (-.5,3.25) {$(4,5)$};
\node at (.5,3.25) {$(2,6)$};
\node at (1.5,3.25) {$(4,1)$};
\node at (2.5,3.25) {$\#$};

\node[blue] at (.5,4.75) {$(4,7)$};
\node[blue] at (-.5,4.75) {$(4,5)$};
\node at (-.5,4.25) {$(2,6)$};
\node at (.5,4.25) {$(4,1)$};
\node at (1.5,4.25) {$\#$};

\node[blue] at (.5,5.75) {$(1,7)$};
\node[blue] at (-.5,5.75) {$(4,5)$};
\node at (-1.5,5.25) {$(2,6)$};
\node at (-.5,5.25) {$(1,4)$};
\node at (.5,5.25) {$\#$};

\node[blue] at (.5,6.75) {$(1,7)$};
\node[blue] at (-.5,6.75) {$(1,5)$};
\node[blue] at (-1.5,6.75) {$(2,6)$};
\node at (-1.5,6.25) {$(1,4)$};
\node at (-.5,6.25) {$\#$};

\node[blue] at (.5,7.75) {$(1,7)$};
\node[blue] at (-.5,7.75) {$(1,5)$};
\node[blue] at (-1.5,7.75) {$(2,6)$};
\node at (-2.5,7.25) {$(1,4)$};
\node at (-1.5,7.25) {$\#$};

\node[blue] at (.5,8.75) {$(1,7)$};
\node[blue] at (-.5,8.75) {$(1,5)$};
\node[blue] at (-1.5,8.75) {$(2,6)$};
\node[green!50!black] at (-2.5,8.25) {$4_1$};

\node[blue] at (.5,9.75) {$(1,7)$};
\node[blue] at (-.5,9.75) {$(1,5)$};
\node[blue] at (-1.5,9.75) {$(2,6)$};
\node[green!50!black] at (-1.5,9.25) {$4_1$};

\node[blue] at (.5,10.75) {$(1,7)$};
\node[blue] at (-.5,10.75) {$(1,5)$};
\node[green!50!black] at (-1.5,10.25) {$6_2$};
\node[green!50!black] at (-.5,10.25) {$4_1$};

\node[blue] at (.5,11.75) {$(1,7)$};
\node[green!70!black] at (-.5,11.75) {$5_1$};
\node[green!50!black] at (-.5,11.25) {$6_2$};
\node[green!50!black] at (.5,11.25) {$4_1$};

\node[green!70!black] at (.5,12.75) {$7_1$};
\node[green!50!black] at (-.5,12.25) {$5_1$};
\node[green!50!black] at (.5,12.25) {$6_2$};
\node[green!50!black] at (1.5,12.25) {$4_1$};

\node[green!70!black] at (.5,13.75) {$7_1$};
\node[green!50!black] at (.5,13.25) {$5_1$};
\node[green!50!black] at (1.5,13.25) {$6_2$};
\node[green!50!black] at (2.5,13.25) {$4_1$};

\node[green!50!black] at (.5,14.25) {$7_1$};
\node[green!50!black] at (1.5,14.25) {$5_1$};
\node[green!50!black] at (2.5,14.25) {$6_2$};
\node[green!50!black] at (3.5,14.25) {$4_1$};

\node[green!50!black] at (1.5,15.25) {$7_1$};
\node[green!50!black] at (2.5,15.25) {$5_1$};
\node[green!50!black] at (3.5,15.25) {$6_2$};
\node[green!50!black] at (4.5,15.25) {$4_1$};
        \end{footnotesize}
      \end{tikzpicture}
      \caption{Marking connected component}
      \label{fig:algo-mark}
    \end{figure} 
\end{proof}

One can note that since sorting cannot be done in linear time for most
sequential computation models (such as Turing machines), most of the
previous examples are known not to belong to any linear class for
sequential computational models.

\section{Constructive properties}\label{sec:cons}

In this section, we study some stability properties of our functional classes.
We also present a speed-up result for the linear
class and show that this class is robust.\\


First let us observe that, by definition of our complexity classes and
their ability to explicitly mark the output sites, we can easily
derive the following chain of inclusions.

\begin{rem}
  Strict real time is included in synchronized real time which is
  included in linear time.
\end{rem}

One usual question is how the functional classes behave with respect to
operations on functions. A first result is that since we can
construct the Cartesian product of cellular automata, those classes
are stable by Cartesian product.

\begin{rem}
  If $f: I^* \to O^*$ and $g: I^* \to {O'}^*$ are computable in strict
  real time (resp. synchronous real time, linear time) so is the
  Cartesian product: $(f,g) : I^* \to O^* \times {O'}^*$ defined by
  $(f,g)(w)=(f(w),g(w))$.  
\end{rem}

Another natural operation is the composition of two functions. In this
case, since the sum of two linear functions is linear, it is not
difficult to see that linear time is closed under composition.

\begin{rem}
  The set of functions computed in linear time is closed under composition.
\end{rem}

Now, let us go to some more technical stuff and look at the
possibility of speed-up.  The speed-up algorithm presented below will also serve to
prove that strict and synchronous linear time are equivalent.

\begin{prop}[Linear acceleration]\label{propLinAcc}
If a function $f$ is computable in time $n+t(n)$ then,
for any positive ratio $r\in \mathbb{Q}^+$, 
$f$ is also computable in time $n+\lceil r\,t(n)\rceil$.
\end{prop}
\begin{proof}
  Let $\mathcal{A}$ be a CA which computes $f$ in time $n+t(n)$.  One
  will construct a CA $\mathcal{B}$ which computes $f$ in time
  $n+r\,t(n)$.  The CA $\mathcal{B}$ will simulate the behavior of
  $\mathcal{A}$ on any input $w$ in achieving a geometric
  transformation of the space-time diagram of $\mathcal{A}$.  It will
  be done in two steps.  First the space-time diagram will be
  compressed in order to speed up the computation. Then a
  decompression will be achieved to retrieve the output.  In both
  steps, we will freely make use of rational coordinates.  To revert
  then to a discrete space-time diagram is the classic trick which
  consists in grouping and keeping some redundant information in each
  integer site.

\paragraph{First step}
To speed up the computation implies to scale the time axis by some
ratio $r$ and, due to neighborhood constraints, to scale likewise the
space axis.  However, a transformation scaling both in time and space
axis is not feasible directly from the initial configuration.
Here it will be achieved by the composition of two intermediate and symmetric directional scalings which fit the specificities of the device.\\
The first transformation takes advantage of the cell $0$ which
receives a border state from its left at step $1$ and so knows that it
is the leftmost active cell and has no information to expect from its
left.  Concretely, the transformation starts from the cell $0$ at time
$1$ and scales the space-time area above the line $t=c+1$ by the ratio
$r$ in the diagonal direction (up-left to down-right).
Figure~\ref{figCompress} depicts the rational representation of the
space-time diagram resulting from this transformation.
\begin{figure}[!htp]
\begin{center}    
  \begin{tikzpicture}[scale=.3,font=\scriptsize\bf,every left delimiter/.style={xshift=2mm},every right delimiter/.style={xshift=-2mm}]
    \def\cc{7}
    \def\hh{15} 

    \foreach \x in {0,...,\cc} 
    \draw[rojo] (\cc-\x,\cc-\x+1) -- (\cc-\x,\hh) -- (\cc,\hh-\x) (\cc-\x,\hh) -- (0,\hh-\cc+\x);
    \foreach \y in {1,...,\cc}
    \draw[rojo] (\cc,\cc+\y) -- (0,\y) (0,2*\y-1) -- (\y-1,\y) (0,2*\y) -- (\y,\y);

    \begin{scope}
      \clip (0,1)-- (0,0) -- (\cc,0) -- (\cc,\cc+1) -- cycle;
      \foreach \x in  {1,...,\cc}
      \draw [verde] (\x,\x) -- (2*\x,0) (\x-1,\x) -- (2*\x-1,0);
    \end{scope}  
    \foreach \x in {0,...,\cc}
    \draw [verde] (\x,0) -- (\x,\x+1) (\cc-\x,0) -- (\cc,\x);
    
    \foreach \x in {0,...,\cc} \foreach \y in {0,...,\x} 
    \draw [verde,fill=verde!61] (\cc-\x+\y,\y) circle (.3cm);
    \foreach \x in {0,...,\cc}
    \draw[fill=green,green] (\x,0) circle (.3cm);

    \foreach \x in {0,...,\cc} 
    \foreach \y in {0,...,\cc}
    \draw [rojo,fill=rojo!61] (\cc-\x,1+\cc+\y) circle (.3cm);
    \foreach \x in {1,...,\cc} 
    \foreach \y in {1,...,\x}
    \draw [rojo,fill=rojo!61] (\cc-\x,\cc-\x+\y) circle (.3cm);

    \begin{scope}[outer sep=1mm]
    \node[anchor=east] at (0,0) {$0$};
    \node[anchor=east] at (0,1) {$1$};
    \node[anchor=east] at (0,\hh) {$n+t(n)$};
    \node[anchor=west] at (\cc,\cc+1) {$n$};
    \node[anchor=north] at (0,-.5) {$0$};
    \node[anchor=north] at (\cc,-.5) {$n-1$};
    \draw[decorate,decoration=brace] (\cc+.5,-.5) -- node[anchor=north, text centered]  {\bm{$w$}} (-.5,-.5);
    \draw[decorate,decoration=brace] (-.5,\hh+.5) -- node[anchor=south, text centered]{\bm{$f(w)$}} (\cc+.5,\hh+.5);
  \end{scope}
  
    \begin{scope}[xshift=15cm]
      \def\aa{2/3}
      \def\bb{1/3}
      \draw[help lines] (0,0) grid (\cc,\hh);

      \foreach \x in {0,...,\cc} 
      \draw[rojo] (\cc-\x,\cc-\x+1) -- (-\x*\aa+\cc*\aa+\hh*\bb-\bb,-\x*\bb+\cc*\bb+\hh*\aa-\aa+1) -- (\cc*\aa+\hh*\bb-\x*\bb-\bb,\cc*\bb+\hh*\aa-\x*\aa-\aa+1)
      (-\x*\aa+\cc*\aa+\hh*\bb-\bb,-\x*\bb+\cc*\bb+\hh*\aa-\aa+1) -- (-\cc*\bb+\hh*\bb+\x*\bb-\bb,\hh*\aa-\cc*\aa+\x*\aa-\aa+1);
      \foreach \y in {1,...,\cc} 
      \draw[rojo] (\cc+\y*\bb-\bb,\cc+\y*\aa-\aa+1) -- (\y*\bb-\bb,\y*\aa-\aa+1) (\y*\bb*2-2*\bb,2*\y*\aa-2*\aa+1) -- (\y-1,\y) (2*\y*\bb-\bb,2*\y*\aa-\aa+1) --  (\y,\y);
      
      \begin{scope}
        \clip (0,1)-- (0,0) -- (\cc,0) -- (\cc,\cc+1) -- cycle;
        \foreach \x in  {1,...,\cc}
        \draw [verde] (\x,\x) -- (2*\x,0) (\x-1,\x) -- (2*\x-1,0);
      \end{scope}  
      \foreach \x in {0,...,\cc}
      \draw [verde] (\x,0) -- (\x,\x+1) (\cc-\x,0) -- (\cc,\x);
      
      \foreach \x in {0,...,\cc} \foreach \y in {0,...,\x} 
      \draw [verde,fill=verde!61] (\cc-\x+\y,\y) circle (.3cm);
      \foreach \x in {0,...,\cc}
      \draw[fill=green,green] (\x,0) circle (.3cm);

      \foreach \x in {0,...,\cc} 
      \foreach \y in {0,...,\cc}
      \draw [rojo,fill=rojo!61] (\cc-\x*\aa+\y*\bb,\cc-\x*\bb+\y*\aa+1) circle (.2cm);
      \foreach \x in {1,...,\cc} 
      \foreach \y in {1,...,\x}
      \draw [rojo,fill=rojo!61] (\cc-\x+\y*\bb-\bb,\cc-\x+\y*\aa-\aa+1) circle (.2cm);
    \begin{scope}[outer sep=1mm]
    \node[anchor=east] at (0,0) {$0$};
    \node[anchor=east] at (0,1) {$1$};
    \node[anchor=east] at (0,\hh) {$n+t(n)$};
    \node[anchor=north] at (0,-.5) {$0$};
    \node[anchor=north] at (\cc,-.5) {$n-1$};
    \draw[decorate,decoration=brace] (\cc+.5,-.5) -- node[anchor=north, text centered]  {\bm{$w$}} (-.5,-.5);
    \draw[decorate,decoration=brace] (\hh*\bb-1,\hh*\aa+.3) -- node[anchor=south, text centered,fill=white,inner sep=0,outer sep=3mm]{\bm{$f(w)$}} (\cc*\aa+\hh*\bb+.5-\bb,\cc*\bb+\hh*\aa+\bb+.8);
  \end{scope}

    \end{scope}
  \end{tikzpicture}
\caption{Scaling in the diagonal direction by a ratio $r=1/2$}\label{figCompress}
\end{center}
\end{figure}
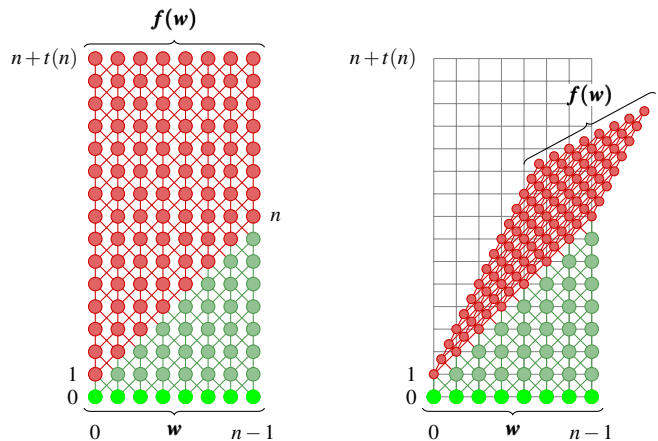

\noindent
The origin of the transformation is the site $\langle 0,1\rangle$ and its matrix is $\left(\begin{array}{cc}(1+r)/2&(1-r)/2\\(1-r)/2&(1+r)/2\end{array}\right)$.
It applies to all sites in the area $\{\langle c,t\rangle\colon 0\leq c <t\}$, i.e., the set of sites impacted by  $\langle 0,1\rangle$ with respect to the neighborhood.
Observe that the transformation is workable.  First, inside the scaled
part, the initial communication links $(-1, 1)$, $(0, 1)$ and $(1, 1)$
are mapped into the links $(-r,r)$, $((1-r)/2,(1+r)/2)$ and $(1, 1)$
which satisfy the dependency constraints.  Second the initial
communication links entering by the down right side of the scaled part
are also mapped into links satisfying the dependency constraints.
Third, no information is coming from its left side and thus leads to
no constraint.
\\\\
In a symmetric way, the second transformation starts from the right
end of the input word at time $1$ and scales the space-time area above
the line $t+c=n$ (with $n$ the size of the input $w$) by the ratio $r$
in the anti-diagonal direction (up-right to down-left).  The
transformation
$\left(\begin{array}{cc}(1+r)/2&-(1-r)/2\\-(1-r)/2&(1+r)/2\end{array}\right)$
with as origin $\langle n-1,1\rangle$ applies to all sites in the area
$\{\langle c,t\rangle\colon 0\leq n-1-c < t\}$.
\\\\
Now, as illustrated in Figure~\ref{figAccLin}, the composition of
these two symmetric scalings sharing the same ratio $r$, results in an
uniform scaling with origin $\langle (n-1)/2,(n+1)/2\rangle$ and ratio
$r$:
$$\left(\begin{array}{cc}(1+r)/2&(1-r)/2\\(1-r)/2&(1+r)/2\end{array}\right)\left(\begin{array}{cc}(1+r)/2&-(1-r)/2\\-(1-r)/2&(1+r)/2\end{array}\right)=\left(\begin{array}{cc}r&0\\0&r\end{array}\right).$$
It applies to the sites $\{\langle c,t\rangle\colon 0\leq c <n \text{ and }
\big|c-(n-1)/2\big|< t-(n-1)/2\}$. 
The uniform scaling arising from two feasible transformations is workable as well while respecting the neighborhood constraints.
\\
Then a simple calculation can confirm that the initial space-time diagram with the output written on cells $0$ to $n-1$ at time $n+t(n)$
is scaled into a diagram where the output is written on the segment of cells
$[(1-r)(n-1)/2,(1+r)(n-1)/2]$ at time $1+(1+r)(n-1)/2+r\,t(n)$.

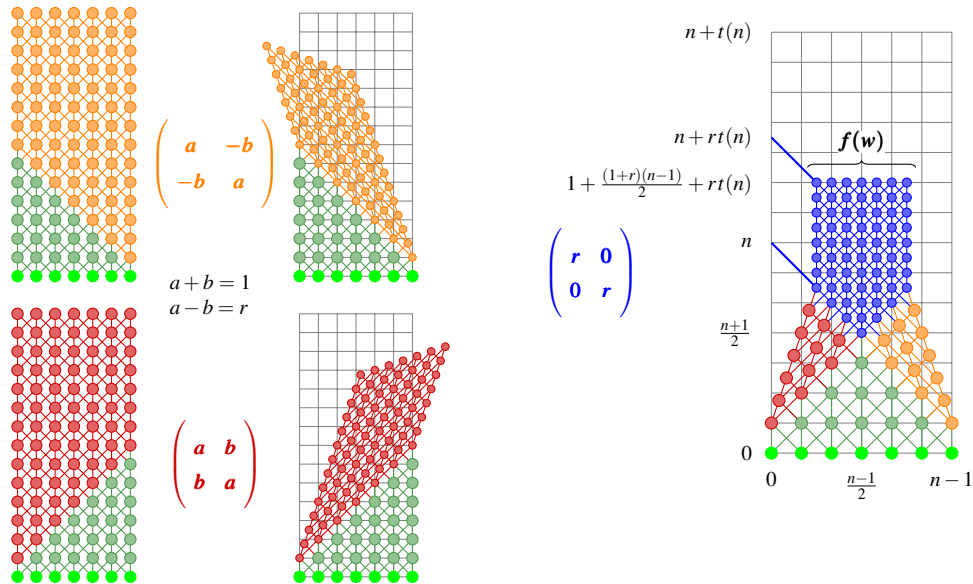
\begin{figure}[!htp]
  \begin{center}
      \def\cc{6}
  \def\hh{13} 
  \begin{tikzpicture}[scale=.25,font=\scriptsize,every left delimiter/.style={xshift=2mm},every right delimiter/.style={xshift=-2mm}]

    \begin{scope}
  \foreach \x in {0,...,\cc} 
  \draw[orange] (\x,\cc-\x) -- (\x,\hh) -- (0,\hh-\x) (\x,\hh) -- (\cc,\hh-\cc+\x) ;
  \foreach \y in {0,...,\cc} 
  \draw[orange] (0,\cc+\y) -- (\cc,\y) (\cc,2*\y-1) -- (\cc-\y,\y-1) (\cc,2*\y) -- (\cc-\y,\y);

  \begin{scope}
\clip (0,\cc-.1)-- (0,-1) -- (\cc,-1) --(\cc,-.1) --  cycle;
\foreach \x in  {0,...,\cc}
\draw [verde] (0,\x) -- (\x+1,-1) (\x,\cc-\x) -- (2*\x-\cc-1,-1) (\x,\cc-\x-1) -- (2*\x-\cc,-1);
\end{scope}  
\foreach \x in  {0,...,\cc}
\draw [verde] (\x,-1) -- (\x,\cc-\x);

\foreach \x in {1,...,\cc} \foreach \y in {1,...,\x} 
\draw [verde,fill=verde!61] (\cc-\x,\y-1) circle (.3cm);
  \foreach \x in {0,...,\cc}
  \draw[fill=green,green] (\x,-1) circle (.3cm);
  \begin{scope}
  \clip (-1,-1) -- (-1,\hh+.4) -- (\cc+1,\hh+.4) -- (\cc+1,-1) -- cycle;
  \foreach \x in {0,...,\cc} 
  \foreach \y in {0,...,\hh}
  \draw [orange,fill=orange!61] (\x,\cc-\x+\y) circle (.3cm);
\end{scope} 
\end{scope}

\begin{scope}[orange,xshift=9.5cm]
  \matrix [matrix of math nodes,left delimiter=(,right delimiter=),outer sep=0,anchor=center] at (1,5) 
  {\bm{a}&\bm{-b}\\\bm{-b}&\bm{a}\\};
\end{scope}

\begin{scope}[xshift=15cm]
  \def\aa{.75}
  \def\bb{.25}
      \draw[help lines] (0,-1) grid (\cc,\hh);

  \begin{scope}
\clip (0,\cc-.1)-- (0,-1) -- (\cc,-1) --(\cc,-.1) --  cycle;
\foreach \x in  {0,...,\cc}
\draw [verde] (0,\x) -- (\x+1,-1) (\x,\cc-\x) -- (2*\x-\cc-1,-1) (\x,\cc-\x-1) -- (2*\x-\cc,-1);
\end{scope}  
\foreach \x in  {0,...,\cc}
\draw [verde] (\x,-1) -- (\x,\cc-\x);

  \foreach \x in {0,...,\cc} 
  \draw[orange] (\x,\cc-\x) -- (\x*\aa+\cc*\bb-\hh*\bb,-\x*\bb+\cc*\bb+\hh*\aa) -- (\cc*\bb-\hh*\bb+\x*\bb,\cc*\bb+\hh*\aa-\x*\aa)
  (\x*\aa+\cc*\bb-\hh*\bb,-\x*\bb+\cc*\bb+\hh*\aa) -- (\cc+\cc*\bb-\hh*\bb-\x*\bb,\hh*\aa-\cc*\aa+\x*\aa);
  \foreach \y in {0,...,\cc} 
  \draw[orange] (-\y*\bb,\cc+\y*\aa) -- (\cc-\y*\bb,\y*\aa) (\cc-\y*\bb*2+\bb,2*\y*\aa-\aa) -- (\cc-\y,\y-1) (\cc-2*\y*\bb,2*\y*\aa) -- (\cc-\y,\y);
  
\foreach \x in {1,...,\cc} \foreach \y in {1,...,\x} 
\draw [verde,fill=verde!61] (\cc-\x,\y-1) circle (.3cm);
  \foreach \x in {0,...,\cc}
  \draw[fill=green,green] (\x,-1) circle (.3cm);
  \begin{scope}
  \clip (-\aa+\bb+\bb*\cc,-\aa+\bb+\bb*\cc) -- (\bb*\cc-\aa-\bb*\hh-\bb*.4,\bb*\cc+\bb+\aa*\hh+\aa*.4) -- (\cc+\aa-\bb*\hh-\bb*.4,-\bb+\aa*\hh+\aa*.4) -- (\cc+1,-1) -- cycle;
  \foreach \x in {0,...,\cc} 
  \foreach \y in {0,...,\hh}
  \draw [orange,fill=orange!61] (\x-\y*\bb,\cc-\x+\y*\aa) circle (.2cm);
\end{scope}

 \end{scope}
\begin{scope}[yshift=-16cm]      

    \begin{scope}[xscale=-1,xshift=-6cm]
  \foreach \x in {0,...,\cc} 
  \draw[rojo] (\x,\cc-\x) -- (\x,\hh) -- (0,\hh-\x) (\x,\hh) -- (\cc,\hh-\cc+\x) ;
  \foreach \y in {0,...,\cc} 
  \draw[rojo] (0,\cc+\y) -- (\cc,\y) (\cc,2*\y-1) -- (\cc-\y,\y-1) (\cc,2*\y) -- (\cc-\y,\y);

  \begin{scope}
\clip (0,\cc-.1)-- (0,-1) -- (\cc,-1) --(\cc,-.1) --  cycle;
\foreach \x in  {0,...,\cc}
\draw [verde] (0,\x) -- (\x+1,-1) (\x,\cc-\x) -- (2*\x-\cc-1,-1) (\x,\cc-\x-1) -- (2*\x-\cc,-1);
\end{scope}  
\foreach \x in  {0,...,\cc}
\draw [verde] (\x,-1) -- (\x,\cc-\x);

\foreach \x in {1,...,\cc} \foreach \y in {1,...,\x} 
\draw [verde,fill=verde!61] (\cc-\x,\y-1) circle (.3cm);
  \foreach \x in {0,...,\cc}
  \draw[fill=green,green] (\x,-1) circle (.3cm);
  \begin{scope}
  \clip (-1,-1) -- (-1,\hh+.4) -- (\cc+1,\hh+.4) -- (\cc+1,-1) -- cycle;
  \foreach \x in {0,...,\cc} 
  \foreach \y in {0,...,\hh}
  \draw [rojo,fill=rojo!61] (\x,\cc-\x+\y) circle (.3cm);
\end{scope} 
\end{scope}

\begin{scope}[rojo,xshift=9.5cm]
  \matrix [matrix of math nodes,left delimiter=(,right delimiter=),outer sep=0,anchor=center] at (1,5) 
  {\bm{a}&\bm{b}\\\bm{b}&\bm{a}\\};
    \node[text width=1.2cm,black] at (1,14) {$a+b=1$ $a-b=r$};
\end{scope}
 
\begin{scope}[xscale=-1,xshift=-21cm]
  \def\aa{.75}
  \def\bb{.25}
      \draw[help lines] (0,-1) grid (\cc,\hh);

  \begin{scope}
\clip (0,\cc-.1)-- (0,-1) -- (\cc,-1) --(\cc,-.1) --  cycle;
\foreach \x in  {0,...,\cc}
\draw [verde] (0,\x) -- (\x+1,-1) (\x,\cc-\x) -- (2*\x-\cc-1,-1) (\x,\cc-\x-1) -- (2*\x-\cc,-1);
\end{scope}  
\foreach \x in  {0,...,\cc}
\draw [verde] (\x,-1) -- (\x,\cc-\x);

  \foreach \x in {0,...,\cc} 
  \draw[rojo] (\x,\cc-\x) -- (\x*\aa+\cc*\bb-\hh*\bb,-\x*\bb+\cc*\bb+\hh*\aa) -- (\cc*\bb-\hh*\bb+\x*\bb,\cc*\bb+\hh*\aa-\x*\aa)
  (\x*\aa+\cc*\bb-\hh*\bb,-\x*\bb+\cc*\bb+\hh*\aa) -- (\cc+\cc*\bb-\hh*\bb-\x*\bb,\hh*\aa-\cc*\aa+\x*\aa);
  \foreach \y in {0,...,\cc} 
  \draw[rojo] (-\y*\bb,\cc+\y*\aa) -- (\cc-\y*\bb,\y*\aa) (\cc-\y*\bb*2+\bb,2*\y*\aa-\aa) -- (\cc-\y,\y-1) (\cc-2*\y*\bb,2*\y*\aa) -- (\cc-\y,\y);
  
\foreach \x in {1,...,\cc} \foreach \y in {1,...,\x} 
\draw [verde,fill=verde!61] (\cc-\x,\y-1) circle (.3cm);
  \foreach \x in {0,...,\cc}
  \draw[fill=green,green] (\x,-1) circle (.3cm);
  \begin{scope}
  \clip (-\aa+\bb+\bb*\cc,-\aa+\bb+\bb*\cc) -- (\bb*\cc-\aa-\bb*\hh-\bb*.4,\bb*\cc+\bb+\aa*\hh+\aa*.4) -- (\cc+\aa-\bb*\hh-\bb*.4,-\bb+\aa*\hh+\aa*.4) -- (\cc+1,-1) -- cycle;
  \foreach \x in {0,...,\cc} 
  \foreach \y in {0,...,\hh}
  \draw [rojo,fill=rojo!61] (\x-\y*\bb,\cc-\x+\y*\aa) circle (.2cm);
\end{scope}

 \end{scope}
\end{scope}
\end{tikzpicture}\hspace{1cm}
\begin{tikzpicture}[scale=.4,font=\scriptsize,every left delimiter/.style={xshift=2mm},every right delimiter/.style={xshift=-2mm}]

  \node at (0,0) {};
  \begin{scope}[blue,yshift=5cm]
  \matrix [matrix of math nodes,left delimiter=(,right delimiter=),outer sep=0] at (1,5) 
  {\bm{r}&\bm{0}\\\bm{0}&\bm{r}\\};
\end{scope}

      \begin{scope}[xshift=7cm,yshift=5cm]
      \def\aa{.75}
      \def\bb{.25}
      \def\acc{.5} 

      \draw[help lines] (0,-1) grid (\cc,\hh);
          \draw[thick,blue] (0,9.5) -- (1.5,8) (0,6) -- (1.5,4.5);
      \begin{scope}[anchor=east,xshift=-3mm]
    \node at (0,\cc) {$n$};
    \node at (0,9.5) {$n+r\,t(n)$};
    \node at (0,-1) {$0$};
    \node at (0,\cc/2) {$\frac{n+1}{2}$};
    \node at (0,8) {$1+\frac{(1+r)(n-1)}{2}+r\,t(n)$};
    \node at (0,\hh) {$n+t(n)$};
  \end{scope}
  \begin{scope}[anchor=north,yshift=-3mm]
    \node at (0,-1) {$0$};
    \node at (\cc,-1) {$n-1$};
    \node at (\cc/2,-1) {$\frac{n-1}{2}$};
  \end{scope}
  \draw[decorate,decoration=brace] (1.2,8.5) -- node[anchor=south, text centered,fill=white,inner sep=0,outer sep=2mm]  {\bm{$f(w)$}} (4.8,8.5);
  \begin{scope}
\clip (0,-.1) --(0,-1) -- (\cc,-1)-- (\cc,-.1) -- (\cc/2,\cc/2-.1) --cycle;
\foreach \x in  {0,...,\cc}
\draw [verde] (\x,-1) -- (\x,\cc/2) (\x,-1) -- (-1,\x) (\x,-1) -- (\cc,\cc-\x-1);
\end{scope}

      \foreach \x in {0,1,2} {
        \draw [rojo] (\x,\x) -- (\x*\acc+\cc*\bb,\cc*\aa-\x*\acc);
        \draw [orange] (\cc-\x,\x) -- (\cc*\aa-\x*\acc,\cc*\aa-\x*\acc); }

      \foreach \y/\d in {0/3,1/2,2/2,3/1,4/1} {
        \draw [rojo] (\y*\bb,\y*\aa) -- (\d+\y*\bb,\y*\aa+\d);
        \draw [orange] (\cc-\y*\bb,\y*\aa) -- (\cc-\d-\y*\bb,\y*\aa+\d);}

      \foreach \y/\d in {1/1.5,2/1,3/2.5,4/2,5/3.5}{
      \draw [rojo] (\y*\bb,\y*\aa) -- (\d*\acc+\y*\bb,\y*\aa-\d*\acc);
      \draw [orange] (\cc-\y*\bb,\y*\aa) -- (\cc-\d*\acc-\y*\bb,\y*\aa-\d*\acc);}

      \foreach \x in {0,...,\cc} 
      \draw[blue] (\cc/4+\x/2,\cc/4+\hh/2) -- (\cc/4,\cc/4+\hh/2-\x/2) (\cc/4+\x/2,\cc/4+\hh/2) -- (\cc/4+\cc/2,\hh/2-\cc/4+\x/2);

      \foreach \y/\d in {6/3,7/4,8/4,9/5,10/5,11/6,12/6}{
      \draw[blue] (\cc/4,\cc/4+\y/2) -- (\cc*\bb-\y*\bb+\d,\cc*\bb+\y*\aa-\d);
      \draw [blue](\cc-\cc/4,\cc/4+\y/2) -- (\cc*\aa+\y*\bb-\d,\cc*\bb+\y*\aa-\d);}

      \foreach \x in {0,...,3}
      \draw[blue] (\cc/4+\x/2,\cc/4+\hh/2) -- (\cc/4+\x/2,\cc/4+\cc/2-\x/2) (\cc-\cc/4-\x/2,\cc/4+\hh/2) -- (\cc-\cc/4-\x/2,\cc/4+\cc/2-\x/2);
  
      \foreach \y/\nb in {0/0,1/1,2/2,3/2,4/1,5/0} 
      \foreach \x in {0,...,\nb} {
        \draw [rojo,fill=rojo!61] (\x*\aa+\y*\bb,\x*\bb+\y*\aa) circle (.2cm);
        \draw [orange,fill=orange!61] (\cc-\x*\aa-\y*\bb,\x*\bb+\y*\aa) circle (.2cm);
      }

\foreach \t/\nb in {0/5,1/3,2/1}
\foreach \d in {1,...,\nb}
\draw [verde,fill=verde!61] (\t+\d,\t) circle (.2cm);
\foreach \t/\nb in {1/1,2/3,3/5,4/7}
       \foreach \d in {1,...,\nb}
       \draw [blue,fill=blue!61] (\cc/2-\t*\acc+\d*\acc,\cc/2+\t*\acc-\acc) circle (.15cm);
      \foreach \x in {0,...,\cc} {
        \foreach \y in {0,...,\cc}
        \draw [blue,fill=blue!61] (\cc/4+\x*\acc,.5+\cc-\cc/4+\y*\acc) circle (.15cm);       
        \draw[fill=green!61,green] (\x,-1) circle (.2cm);
      }

    \end{scope}
  \end{tikzpicture}
\caption{Uniform scaling  as composition of two symmetric scalings with a ratio $r=1/2$}\label{figAccLin}
\end{center}
\end{figure}

\paragraph{Second step}
Finally, the compressed output has to be decompressed.
The geometric construction, depicted in Figure~\ref{figDecomp},
makes use of two signals initiated from each extremity of the compressed output
and run respectively to the left and to the right with a slope
$1/r-1$. The decompression  is halted by a modified firing squad which
marks times $n + rt(n)$.

\begin{figure}[!htp]
  \begin{center}
    \begin{tikzpicture}[scale=.3]
  \def\k{3}
  \def\p{1} 
  \def\nc{10} 
  \def\mil{4} 

  \coordinate (f) at (\k*\nc/2-\k/2,0);
  \coordinate (a) at (0,-\nc*\p+\p);
  \coordinate (c) at (\nc*\p+\nc-\p-1,-\nc*\p+\p);

  \draw[help lines,xstep=3cm,ystep=3cm] (0,-\nc*\p+\p) grid (\k*\nc-\k,0);
  \draw[help lines] (0,-\nc*\p+\p) rectangle (\k*\nc-\k,0);

  \foreach \i in {0,...,\mil}{
    \coordinate (e\i) at (\nc*\p-\p+\i,-\nc*\p+\p);
    \coordinate (s\i) at (\k*\i,0);
    \coordinate (d\i) at (\nc*\p-\p+\i,-\nc*\p+\p+\k*\i-\i);
    \draw (e\i) -- (d\i) -- (s\i);
  }
  \draw[line width=1mm,azul] (0,0) -- (\k*\nc-\k,0) (e0) -- (c);
  \draw[line width=1mm,rojo] (s0) -- (\k*\nc-\k,0);
  \draw[decorate,decoration=brace] (0,.5) -- node[anchor=south, text centered]  {\bm{$f(w)$}} (\k*\nc-\k,.5);
  \draw[decorate,decoration=brace] (\nc*\p+\nc-\p-1,-\nc*\p+\p-.5) -- node[anchor=north, text centered]{\bm{$f(w)$}} (\nc*\p-\p,-\nc*\p+\p-.5);
  
  \draw[ultra thick,violet] (e0) -- (f);
  \draw[ultra thick,orange] (e0) -- (s0);

  \begin{scope}[outer sep=1mm]
    \node[anchor=east] at (s0) {$n+r\,t(n)$};
    \node[anchor=east] at (a) {$1+\frac{(1+r)(n-1)}{2}+r\,t(n)$};
    \node[anchor=north,yshift=-3mm] at (e0) {$\frac{(1-r)(n-1)}{2}$};
    \node[anchor=north,yshift=-3mm] at (c) {$\frac{(1+r)(n-1)}{2}$};
  \end{scope}
  \begin{scope}[xscale=-1,xshift=-27 cm]
    \coordinate (f) at (\k*\nc/2-\k/2,0);
      \foreach \i in {0,...,\mil}{
    \coordinate (e\i) at (\nc*\p-\p+\i,-\nc*\p+\p);
    \coordinate (s\i) at (\k*\i,0);
    \coordinate (d\i) at (\nc*\p-\p+\i,-\nc*\p+\p+\k*\i-\i);
    \draw (e\i) -- (d\i) -- (s\i);
  }
  \draw[ultra thick,violet] (e0) -- (f);
  \draw[ultra thick,orange] (e0) -- (s0);
\end{scope}
\end{tikzpicture}
\caption{Decompression step (with as ratio $r=1/3$)}\label{figDecomp}
\end{center}
\end{figure}
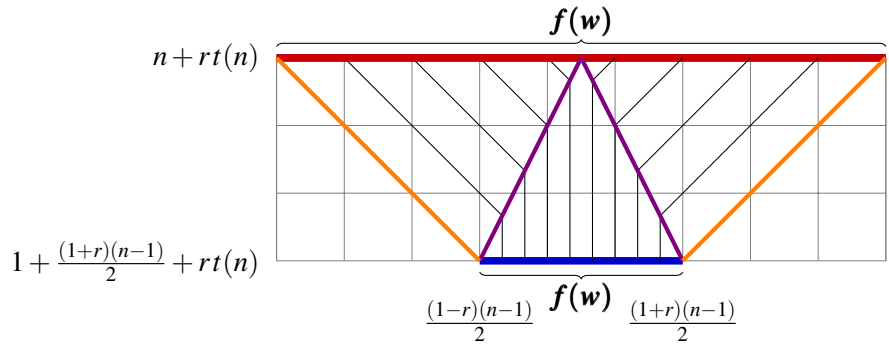
\end{proof}
\bigskip

In fact, when looking more in detail at this proof, it can be seen
that the result applies not only in (synchronous) linear time but
also in the linear extension of strict real time. 

\begin{prop}\label{lem:lin_equiv}
  The linear extension of the strict real time is equivalent to linear time.
\end{prop}

To get a comprehensive picture of functional classes, one main
question is to determine the power of synchronous real time, notably
in comparison with synchronous linear time. Unfortunately we fail to
give a concrete answer.  This is not surprising, we face the same
difficulty to separate real time and linear time in the case of
recognition CA and even though the equality seems unlikely.  The next
proposition makes explicit the links between the functional and
recognition contexts regarding to real time ability.

\begin{prop}
  The following assertions are equivalent:
  \begin{enumerate}
    \item The composition of any two strict real time functions is computable in synchronous real time.
    \item The synchronous real time functional class is closed under composition.
    \item The synchronous real time and synchronous linear time functional classes are equivalent.
    \item The real time and linear time recognition classes are equivalent.
    \item The real time recognition class is closed under reversal.
    \end{enumerate}
  \end{prop}
  \begin{proof}
    The equivalence $(4) \Leftrightarrow (5)$ is far to be trivial: it is a meaningful result established by O.~H.~Ibarra and T.~Jiang~\cite{IbarraJiang88}.
The implication  $(3) \Rightarrow (2)$ is straightforward as linear time functional class is closed under composition. The implication  $(2) \Rightarrow (1)$ is
immediate.
The implications $(1) \Rightarrow (5)$ and $(4) \Rightarrow (3)$  are presented below in Lemma~\ref{lem1to5} and Lemma~\ref{lem4to3}.
\end{proof}

\begin{lem}\label{lem1to5}
  If the composition of any two strict real time functions is computable in synchronous real time then the real time recognition class is closed under reversal.
\end{lem}
\begin{proof}
  Let $L$ be a language recognized in real time and its reverse denoted by $L^R$ . According to Lemma~\ref{lemRec2Fnl}, a CA $\mathcal{A}$ computes in strict real time the function $f$ defined by $f(w)=10^{|w|-1}$ if $w\in L$ and $f(w)=0^{|w|}$ otherwise. And its mirror CA  computes in strict real time the function $g$ defined by $g(w)=0^{|w|-1}1$ if $w\in L^R$ and $g(w)=0^{|w|}$ otherwise.
  Independently  the function $h$ defined by $h(x_1,\cdots,x_n)=x_n0^{n-1}$ is also computable in strict real time. 
  The composition of $h$ and $g$ corresponds to the function $i=h\circ g$ defined by $i(w)=10^{|w|-1}$ if $w\in L^R$ and $i(w)=0^{|w|}$ otherwise.
  If the hypothesis holds then $i$ is computable in synchronous real time and it follows that $L^R$ is recognized in real time. 
\end{proof}

\begin{lem}\label{lem4to3}
  If the real time and linear time recognition classes are equivalent then the synchronous real time and synchronous linear time functional classes are equivalent.
\end{lem}

\begin{proof}
The proof proceeds in two steps. First, making use of the hypothesis, we will show that if a function $f$ is computable in linear time then there exists a CA which for any input $w$ yields every $i$-th output symbol $f(w)_i$ on the site $\langle 0,|w|+i\rangle$. 
Second, a compression will allow to get the first half of the output at synchronous real time.
As regards the second half of the output, these two steps will be applied in a symmetric way.
 
\paragraph{First step}
Let $f$ be a function computable in linear time.
So there exists a CA $\mathcal{A}$ which on every input $w$ outputs $f(w)$
at time $2|w|$ (See Figure~\ref{figlqlinrec}(a)).
In addition, the CA $\mathcal{A}$ can be completed in order to get 
the output $f(w)$ on the leftmost cell between times $2|w|$ and $4|w|$:
from the time  $2|w|$, each output symbol has to be sent at speed $1/2$ towards the leftmost cell (See Figure~\ref{figlqlinrec}(b)). In particular the $i$-th output symbol $f(w)_i$ reaches the leftmost cell at time $2|w|+2i-2$.
Now, let us translate what does it mean in term of recognition capacity.
We introduce a padding symbol $\sharp$  not belonging to the input alphabet $I$.
For any symbol $q$ of the output alphabet, consider the language $L_q=\{w\sharp^i\colon f(w)_i=q\}$.
Then the CA $\mathcal{A}$ provides the basic ingredients to recognize the language $L_q$ and that in linear time (See Figure~\ref{figlqlinrec}(c)). 

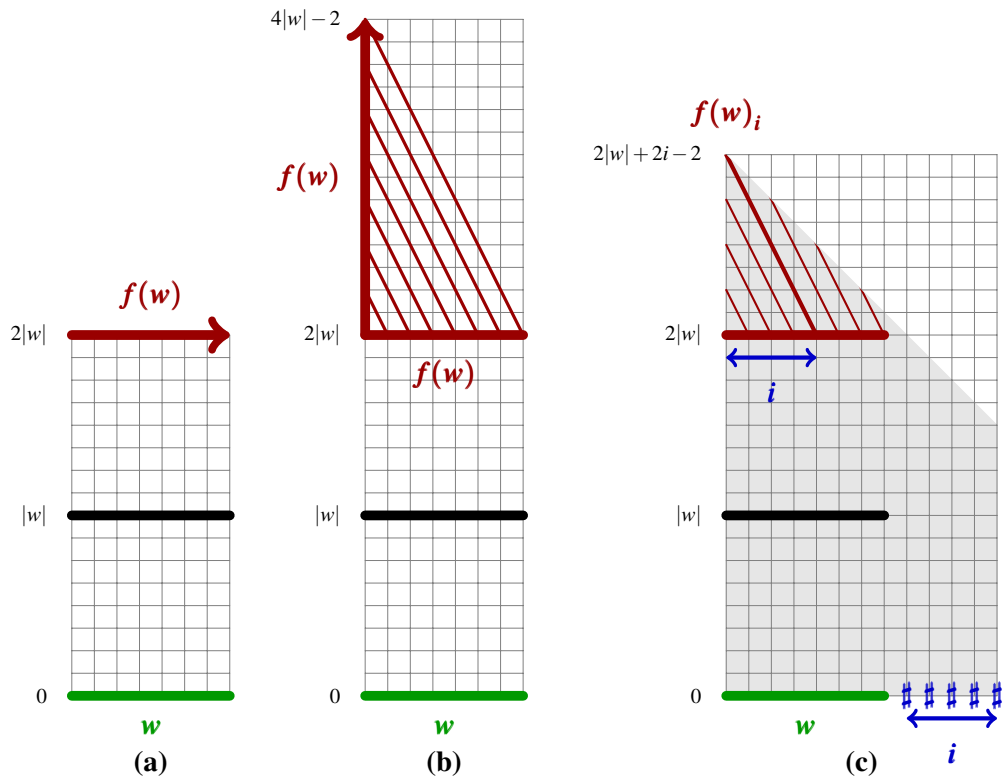
\begin{figure}[!htp]
\begin{center}
      \begin{tikzpicture}[scale=.3,cap=round,line width=1.3mm,text centered,outer sep=2mm]
    \def\nb{8}
    \draw[help lines] (1,0) grid (\nb,2*\nb);
    \draw[incol] (1,0) -- node[anchor=north]{\normalsize $\bm{w}$} (\nb,0);
    \draw (1,\nb) --  (\nb,\nb);
    \draw[outcol,->] (1,2*\nb) -- node[anchor=south]  {\normalsize $\bm{f(w)}$} (\nb,2*\nb);
    \node at (\nb/2+.5, -3) {\bf (a)};

    \begin{scope}[anchor=east,font=\scriptsize]
      \node at (1,0) {$0$};
      \node at (1,\nb) {$|w|$};
      \node at (1,2*\nb) {$2|w|$};
    \end{scope}

    \begin{scope}[xshift=13cm]
    \draw[help lines] (1,0) grid (\nb,4*\nb-2);
    \draw[incol] (1,0) -- node[anchor=north]{ $\bm{w}$} (\nb,0);
    \draw (1,\nb) -- (\nb,\nb);
    \draw[outcol] (1,2*\nb) -- node[anchor=north] {$\bm{f(w)}$} (\nb,2*\nb);
    \draw[outcol,->] (1,2*\nb) -- node[anchor=east] {$\bm{f(w)}$} (1,4*\nb-2);

    \foreach \c in {1,...,\nb}
    \draw[outcol,very thick] (\c,2*\nb) -- (1,2*\nb+2*\c-2);
    \node at (\nb/2+.5, -3) {\bf (b)};

    \begin{scope}[anchor=east,font=\scriptsize]
      \node at (1,0) {$0$};
      \node at (1,\nb) {$|w|$};
      \node at (1,2*\nb) {$2|w|$};
      \node at (1,4*\nb-2) {$4|w|-2$};
    \end{scope}

    \end{scope}
    
    \begin{scope}[xshift=29cm]
      \def\nbi{5}
      \draw[fill=gray!20,gray!20,line width=0] (1,0) -- (1,2*\nb+2*\nbi-2) -- (\nb+\nbi,\nb+\nbi-1) -- (\nb+\nbi,0) -- cycle;
      \draw[help lines] (1,0) grid (\nb+\nbi,2*\nb+2*\nbi-2);
      \begin{scope}
        \clip (0,0) --  (0,2*\nb+2*\nbi-2) -- (1,2*\nb+2*\nbi-2) -- (\nb+\nbi,\nb+\nbi-1) -- (\nb+\nbi,0) -- cycle;
        \foreach \c in {1,...,\nb}
        \draw[outcol,thick] (\c,2*\nb) -- (1,2*\nb+2*\c-2);
        \draw[outcol,ultra thick] (\nbi,2*\nb) -- (1,2*\nb+2*\nbi-2);
        \draw[outcol] (1,2*\nb) --  (\nb,2*\nb);
      \end{scope}
      \node at (\nb/2+\nbi/2+.5, -3) {\bf (c)};
      \foreach \c in {1,...,\nbi}
      \node[azul] at (\nb+\c,0) { $\bm{\sharp}$};
      \draw[azul,<->,yshift=-1cm,ultra thick] (\nb+1,0) -- node[anchor=north]{$\bm{i}$} (\nb+\nbi,0);

      \draw[azul,<->,yshift=-1cm,ultra thick] (1,2*\nb) -- node[anchor=north]{$\bm{i}$} (\nbi,2*\nb);

      \draw[incol] (1,0) -- node[anchor=north]{ $\bm{w}$} (\nb,0);
      \draw (1,\nb) -- (\nb,\nb);
      \node[outcol,anchor=south] at (1,2*\nb+2*\nbi-2) {$\bm{f(w)_i}$}; 

    \begin{scope}[anchor=east,font=\scriptsize]
      \node at (1,0) {$0$};
      \node at (1,\nb) {$|w|$};
      \node at (1,2*\nb) {$2|w|$};
      \node at (1,2*\nb+2*\nbi-2) {$2|w|+2i-2$};
    \end{scope}

    \end{scope}

  \end{tikzpicture}
  \caption{The CA $\mathcal{A}$}\label{figlqlinrec}
\end{center}
\end{figure}

\noindent
Next, the assumption that real time and linear time are equivalent in the recognition framework, 
allows one to deduce the existence of a CA $\mathcal{A}_q$
which recognizes the language $L_q$ in real time.
Then consider the CA $\mathcal{B}$ resulting from the cross product of all  the CA $\mathcal{A}_q$ where $q$ ranges over the output alphabet. 
Such a CA $\mathcal{B}$ has everything it needs to produce,
on input $w\sharp^i$ the $i$-th output symbol $f(w)_i$ at time $|w|+i-1$.
See Figure~\ref{figlqrtrec}.
Moreover observe that the running of $\mathcal{B}$ 
on the input word $w\sharp^{|w|}$ contains all the real-time evolutions 
on the prefixes of  $w\sharp^{|w|}$. That means $\mathcal{B}$ is able to produce the whole output $f(w)$ on the cell $0$ between times $|w|$ and $2|w|$.
Lastly, in order to $\mathcal{B}$ works in the same way on the input
$w$ clear of padding symbol, just fold the space-time diagram along
the vertical line $c=|w| -1/2$.

\begin{figure}
\begin{center}
  \begin{tikzpicture}[scale=.3,cap=round,line width=1.3mm,text centered,outer sep=2mm]
    \def\nb{8}
    \def\nbi{5}
    \draw[fill=gray!20,gray!20,line width=0] (1,0) -- (1,\nb+\nbi-1)  -- (\nb+\nbi,0) -- cycle;
    \draw[help lines] (1,0) grid (\nb+\nbi,\nb+\nbi-1);
    \foreach \c in {1,...,\nbi}
    \node[azul] at (\nb+\c,0) {$\bm{\sharp}$};
    \draw[azul,<->,yshift=-1cm,ultra thick] (\nb+1,0) -- node[anchor=north]{$\bm{i}$} (\nb+\nbi,0);

    \draw[incol] (1,0) -- node[anchor=north]{$\bm{w}$} (\nb,0);
    \node[outcol,anchor=south] at (1,\nb+\nbi-1) {$\bm{f(w)_i}$}; 

    \begin{scope}[anchor=east,font=\scriptsize]
      \node at (1,0) {$0$};
      \node at (1,\nb) {$|w|$};
      \node at (1,\nb+\nbi-1) {$|w|+i-1$};
    \end{scope}

    \begin{scope}[xshift=18cm]
    \draw[help lines] (1,0) grid (2*\nb,2*\nb-1);
    \draw[azul,<->,yshift=-1cm,ultra thick] (\nb+1,0) -- node[anchor=north]{$\bm{|w|}$} (2*\nb,0);

\foreach \c in {1,...,\nb}{
    \draw[gray,semithick] (\nb+\c,0) -- (1,\nb+\c-1);
 \node[azul] at (\nb+\c,0) {$\bm{\sharp}$};
}
    \draw[incol] (1,0) -- node[anchor=north]{$\bm{w}$} (\nb,0);
    \draw[outcol,->] (1,\nb) -- node[anchor=east] {$\bm{f(w)}$} (1,2*\nb-1);

    \begin{scope}[anchor=east,font=\scriptsize]
      \node at (1,0) {$0$};
      \node at (1,\nb) {$|w|$};
      \node at (1,2*\nb-1) {$2|w|-1$};
    \end{scope}

    \draw[dashed,thin,red] (\nb+.5,0) -- +(0,2*\nb-1);
    \node[anchor=south,font=\tiny] at (\nb+.5,2*\nb-2) {$|w|-1/2$};
  \end{scope}

\end{tikzpicture}
  \caption{The CA $\mathcal{B}$}\label{figlqrtrec}
\end{center}
\end{figure}
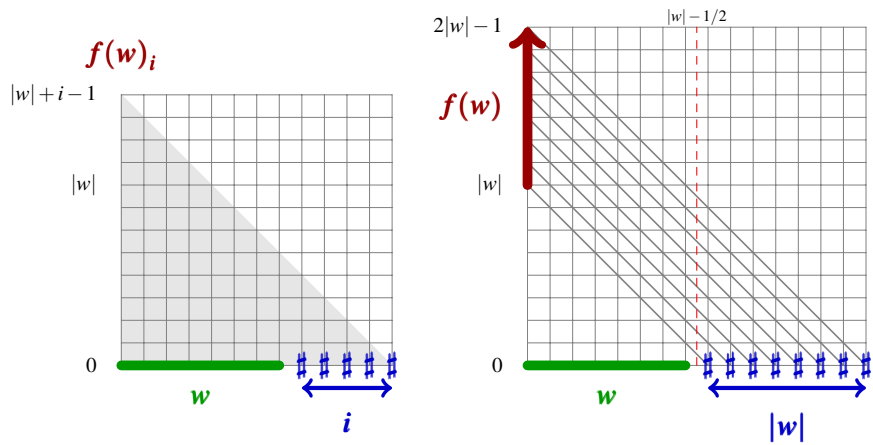

\paragraph{Second step}
Now we will retrieve the first half of the output in synchronous real time.
For this purpose, we will apply the compression presented in Proposition~\ref{propLinAcc} on the CA $\mathcal{B}$, and thereby move the segment $[\langle 0,|w|\rangle,\cdots,\langle 0,|w|+(|w|-1)/2\rangle]$ which yields the first half of the output 
to the segment $[\langle 0,|w|\rangle,\cdots,\langle (|w|-1)/2,|w|\rangle]$.
Precisely, making use of the transformation $\left(\begin{array}{cc}2/3&1/3\\1/3&2/3\end{array}\right)$ with origin  $\langle 0,1\rangle$, 
each of the site $\langle 0,|w|+i\rangle$ carrying the $i$-th output symbol
is mapped to the site $\langle (|w|-1+i)/3,1+2(|w|-1+i)/3\rangle$. See Figure~\ref{fighalf}(b). Next, as depicted in Figure~\ref{fighalf}(c), the data are transmitted at maximal speed to the left. Thus, for all indices $i$ into the first half, i.e. $i\leq (|w|-1)/2$, the  $i$-th output symbol reaches after $(|w|-1-2i)/3$ steps the site $\langle i,|w|\rangle$.\\

\begin{figure}
\begin{center}
  \begin{tikzpicture}[scale=.3,cap=round,line width=1mm,text centered,outer sep=2mm]
    \def\nb{8}
    \def\demi{4}
    \draw[help lines] (1,0) grid (\nb,2*\nb-1);

    \draw[fill=incol!20,incol!20,very thin] (1,0) --  (\nb,0) -- (\nb,\nb-1) --cycle;
    \draw[fill=orange!20,orange!20,very thin] (1,1) -- (1,2*\nb-1) -- (\nb,2*\nb-1) -- (\nb,\nb) --cycle;
    \draw[help lines] (1,0) grid (\nb,2*\nb-1);

    \draw[incol] (1,0) -- node[anchor=north]{ $\bm{w}$} (\nb,0);
    \draw[outcol,->] (1,\nb) -- node[anchor=east] {$\bm{f(w)}$} (1,2*\nb-1);
    \node at (\nb/2+.5,-3) {\bf (a)};

    \foreach \c in {1,...,\nb} 
      \draw [outcol,fill=outcol] (1,\nb+\c-1) circle (.1cm);

    \begin{scope}[anchor=east,font=\scriptsize]
      \node at (1,0) {$0$};
      \node at (1,\nb) {$|w|$};
      \node at (1,2*\nb-1) {$2|w|-1$};
    \end{scope}

    \begin{scope}[xshift=14cm]
      \draw[help lines] (1,0) grid (\nb,2*\nb-1);

      \coordinate (a) at (2/3+\nb/3,1/3+2*\nb/3);
      \coordinate (b) at (1/3+2*\nb/3,-1/3+4*\nb/3);
      \coordinate (c) at (-1/3+4*\nb/3,-2/3+5*\nb/3);

    \draw[fill=incol!20,incol!20,very thin] (1,0) --  (\nb,0) -- (\nb,\nb-1) --cycle;
    \draw[fill=orange!20,orange!20,very thin] (1,1) -- (b) -- (c) -- (\nb,\nb) --cycle;
    \draw[help lines] (1,0) grid (\nb,2*\nb-1);

    \draw[incol] (1,0) -- node[anchor=north]{ $\bm{w}$} (\nb,0);
    \draw[outcol,->] (a) -- node[anchor=west] {$\bm{f(w)}$} (b);
    \node at (\nb/2+.5,-3) {\bf (b)};

    \foreach \c in {1,...,\nb} {
      \coordinate (d\c) at (\nb/3+\c/3+1/3,2*\nb/3+2*\c/3-1/3);
      \draw [outcol,fill=outcol] (d\c)   circle (.1cm);
    }

    \begin{scope}[anchor=east,font=\scriptsize]
      \node at (1,0) {$0$};
      \node at (1,\nb) {$|w|$};
      \node at (1,2*\nb-1) {$2|w|-1$};
    \end{scope}
  \end{scope}
    \begin{scope}[xshift=28cm]
      \draw[help lines] (1,0) grid (\nb,2*\nb-1);

      \coordinate (a) at (2/3+\nb/3,1/3+2*\nb/3);
      \coordinate (b) at (1/3+2*\nb/3,-1/3+4*\nb/3);
      \coordinate (c) at (-1/3+4*\nb/3,-2/3+5*\nb/3);

    \draw[fill=incol!20,incol!20,very thin] (1,0) --  (\nb,0) -- (\nb,\nb-1) --cycle;
    \draw[fill=orange!20,orange!20,very thin] (1,1) -- (b) -- (c) -- (\nb,\nb) --cycle;
    \draw[help lines] (1,0) grid (\nb,2*\nb-1);

    \draw[incol] (1,0) -- node[anchor=north]{ $\bm{w}$} (\nb,0);
    \draw[outcol,->] (a) -- node[anchor=south east] {$\bm{f(w)}$} (b);
    \node at (\nb/2+.5,-3) {\bf (c)};

    \foreach \c in {1,...,\nb} {
      \coordinate (d\c) at (\nb/3+\c/3+1/3,2*\nb/3+2*\c/3-1/3);
      \draw [outcol,fill=outcol] (d\c)   circle (.1cm);
    }
    \foreach \c in {1,...,\demi} {
      \coordinate  (e\c) at (\c,\nb);
      \draw [outcol,fill=outcol] (e\c) circle (.1cm);
      \draw [outcol] (d\c) -- (e\c);
    }

    \begin{scope}[anchor=east,font=\scriptsize]
      \node at (1,0) {$0$};
      \node at (1,\nb) {$|w|$};
      \node at (1,2*\nb-1) {$2|w|-1$};
    \end{scope}
  \end{scope}
\end{tikzpicture}
\caption{Transfer of the first half of the output}\label{fighalf}
\end{center}
\end{figure}
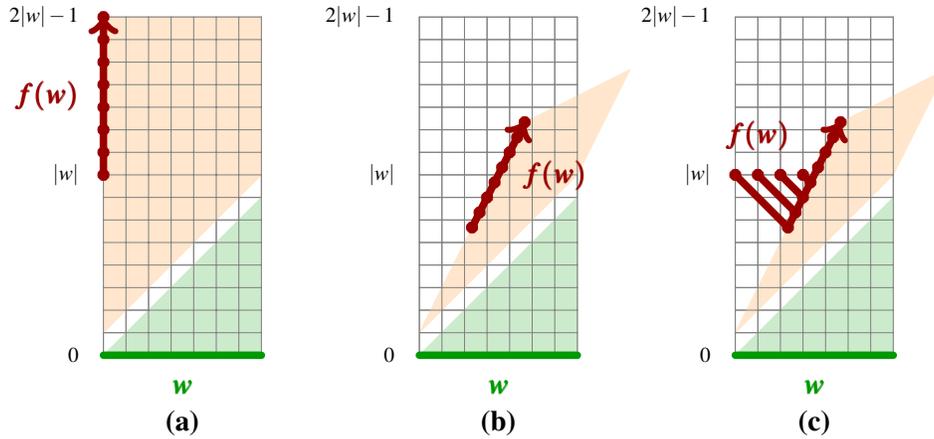

\noindent
Finally, the two stages are achieved in a symmetric way (i.e., making use of a vertical symmetry axis) to reconstitute the second half of the output.

\end{proof}

\section{Limitations}\label{sec:hier}

Achieving negative results is often a more difficult task. Here, we
shall present a result which makes explicit the impossibility of some information
transfer in cellular automata. To do this we shall prove a lack of
certain patterns in the space-time diagram.

To be more precise, if we look at a triangular extract from a
space-time diagram, we know that states inside the \emph{triangle} of
size $n$ are uniquely determined by its \emph{base} $b=b_1b_2\ldots
b_n$ . The \emph{height} of the triangle is of $h = \lceil n/2 \rceil$.

Now, let us take any arbitrary partition of the states set into two
subsets $\pi : S \to \{0,1\}$.

\begin{defn}
  A triangular extract of a space-time diagram is said \emph{uniform}
  (with respect to a partition $\pi$) if for any row of the triangle,
  all states in the row are in the same part. 
\end{defn}

We can thus associate to any uniform triangle of size $n$ a
\emph{characteristic word} $c = c_1c_2 \ldots c_{\lceil n/2 \rceil}
\in \{0,1\}^{\lceil n/2 \rceil}$ corresponding to the sequence of the
parts attached to each row (see Figure~\ref{fig:induc}).

\begin{thm}
  For any partition $\pi$, there exists a characteristic word $c \in \{0,1\}^*$ such
  that no uniform triangle with this characteristic exists. 
\end{thm}

\begin{proof}\label{thm:impo}
  Let us pose the sequence $v_n$ defined by :
  \[ v_1 = \frac{|S|^4}{4} \textrm{ and } v_{n+1} =
  \frac{v_n^2}{2^{2^{n-1}}}\]

  To show the theorem, let us prove that for uniform triangle of size
  $2^{n+1}$, there exists a characteristic word $c_n$ of size $2^n$
  generated by at most $v_n$ bases by recurrence over $n$. As $v_n$
  converges towards $0$, this leads to the desired result.
 
  For $n=1$, a uniform triangle of height $2$ has a base of size
  $4$. There are $|S|^4$ different bases, $2^2=4$ different
  characteristics. There exists a characteristic word $c_0 \in
  \{0,1\}^2$ which has at most $|S|^4/4$ corresponding bases.

  Let us take a uniform triangle of size $2^{n+1}$ and divide it into
  four smallest ones as depicted in Figure~\ref{fig:induc}. By
  definition, $A$ and $B$ correspond to uniform triangles of size
  $2^n$ sharing the same characteristic.  By recurrence, there exists
  a characteristic word $c_n$ such that the number of different bases
  $b_A$ and $b_B$ are both less than $v_n$. Since $b_A b_B$ is the
  base of our triangle, we can deduce that any uniform triangle of
  size $2^{n+1}$ whose characteristic begins with $c_n$ must have its
  base among the $v_n^2$ elements on the form $b_Ab_B$. Since the
  number of characteristic words of size $2^n$ with prefix $c_n$ (of
  length $2^{n-1}$) is $2^{2^{n-1}}$, the average number of base per
  words is $v_{n+1}$ and at least one word has thus less or equal
  bases than $v_{n+1}$

  \begin{figure}[!htp]
    \centering
    \begin{tikzpicture}[auto,scale=.45]
      \fill[black!20] (0,0) rectangle +(14,1);
\fill[black!40] (1,1) rectangle +(12,1);
\fill[black!20] (2,2) rectangle +(10,1);
\fill[black!15] (3,3) rectangle +(8,1);
\fill[black!40] (4,4) rectangle +(6,1);
\fill[black!20] (5,5) rectangle +(4,1);
\fill[black!15] (6,6) rectangle +(2,1);

\foreach \i in {0,1,...,6}
  \draw (\i,\i) grid +(14-\i-\i,1);

\node at (.5,.5) {$b_1$};
\node at (1.5,.5) {$b_2$};
\node at (2.5,.5) {$b_3$};
\draw[dotted] (3.5,.5) -- (11.5,.5);
\node at (12.5,.5) {$b_{13}$};
\node at (13.5,.5) {$b_{14}$};

\node at (14.5,.5) {$c_1$};
\node at (14.5,1.5) {$c_2$};
\draw[dotted] (14.5,2.5) -- (14.5,4.5);
\node at (14.5,5.5) {$c_6$};
\node at (14.5,6.5) {$c_7$};
    \end{tikzpicture} \hspace{1.5cm} %
    \begin{tikzpicture}[scale=.45,auto]
      \fill[black!10] (1,1) -- (1.5,1.5) -- (10.5,1.5) -- (11,1);
\fill[black!15] (2,2) -- (2.5,2.5) -- (9.5,2.5) -- (10,2);
\fill[black!10] (2.5,2.5) -- (3,3) -- (9,3) -- (9.5,2.5);
\fill[black!15] (3.5,3.5) -- (4.5,4.5) -- (7.5,4.5) -- (8.5,3.5);
\fill[black!10] (4.5,4.5) -- (5,5) -- (7,5) -- (7.5,4.5);

\draw[thick] (0,0) -- (6,6) -- (12,0) -- cycle;

\draw (3,3) -- (9,3) -- (6,0) -- cycle;

\draw[blue,thick,dotted] (0,0) -- (3,3) -- (6,0) -- cycle;
\draw[blue,thick,dotted] (6,0) -- (9,3) -- (12,0) -- cycle;

\node[blue] at (3,1.5) {$\bm{A}$};
\node[blue] at (9,1.5) {$\bm{B}$};

\draw[decorate,decoration=brace] (5.9,-.5) -- node[anchor=north, text centered]{\bm{$b_A$}} (0,-.5);
\draw[decorate,decoration=brace] (12,-.5) -- node[anchor=north, text centered]{\bm{$b_B$}} (6.1,-.5);

\draw[decorate,decoration=brace] (12.5,3) -- node[anchor=west]{$c_n$} (12.5,0);
\draw[decorate,decoration=brace] (13.5,6) -- node[anchor=west]{$c_{n+1}$} (13.5,0);
    \end{tikzpicture}
    
    \caption{Induction over uniform triangle}
    \label{fig:induc}
  \end{figure}
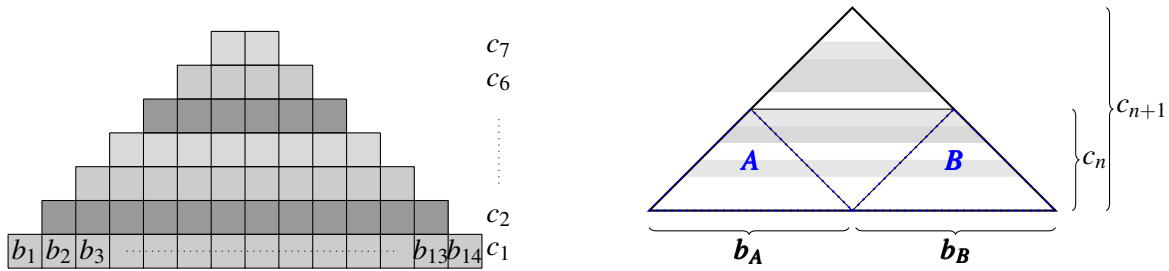
  
\end{proof}

Using this general result, we can prove that the mirror is not
computable in strict real time.

\begin{prop}
  Mirror is not computable is strict real time.
\end{prop}

\begin{proof}
  By contradiction, assume that there exists a CA computing such a
  solution. Such CA  works over each prefix of the input as depicted in
  the first row of Figure~\ref{fig:trian-mirror} and output the result
  at the specified position.

  Since the $\#$ symbol cannot influence any cell under the diagonal, 
  the computation in the lower left triangle is the same for all
  cases. Therefore, if we use the classical trick of adding to each
  cell an additional  layer which makes computation as if the input
  was finished, the resulting cellular automaton can do the partial
  computation indicated on the lower row of the figure.
  
  Taking the projection $\pi$ used for outputting the result, we can
  see that we could construct uniform triangles with any arbitrary
  characteristics contradicting the previous theorem.

  \begin{figure}[!htp]
    \centering
    \begin{tikzpicture}[scale=.3,font=\scriptsize,anchor=center]
  
  \begin{scope}
   \def\nb{10}
    \def\demid{6}

    \fill[yellow!25] (1,0) -- (7,0) -- (1,6) -- cycle; 

    \foreach \c in {1,...,\nb}
    \coordinate[label=below:$x_{\c}$]  (e\c) at (\c,0);

    \foreach \c in {\demid,...,\nb} {
      \coordinate[label=above:$x_{\c}$] (s\c) at (\nb-\c+1,\c);
      \coordinate (g\c) at (\demid-\c/2-.5,3*\c/2-\demid+.5);
      \coordinate (d\c) at (\demid+\c/2-.5,\demid-\c/2-.5);
    }

    \draw[azul,ultra thick] (e7) -- (g7) -- (s7); 

    \draw[lightgray] (1,0) grid (\nb,\nb);
    \draw[verde,ultra thick] (1,0) -- (\nb,0);
    \draw [rojo,very thick] (1,1) -- (\nb,\nb) (\nb,1) -- (1,\nb);

    \foreach \c in {\demid,...,\nb} 
    \draw[azul] (e\c) -- (g\c) -- (s\c);
    \node at (0,0) {$\sharp$}; \node at (\nb+1,0) {$\sharp$};
  \end{scope}

\begin{scope}[xshift=13cm]
   \def\nb{9}
   \def\demid{5.5}
   \def\cem{6}

   \fill[yellow!25] (1,0) -- (7,0) -- (1,6) -- cycle; 

    \foreach \c in {1,...,\nb}
    \coordinate[label=below:$x_{\c}$]  (e\c) at (\c,0);

    \foreach \c in {\cem,...,\nb} {
      \coordinate[label=above:$x_{\c}$] (s\c) at (\nb-\c+1,\c);
      \coordinate (g\c) at (\demid-\c/2-.5,3*\c/2-\demid+.5);
      \coordinate (d\c) at (\demid+\c/2-.5,\demid-\c/2-.5);
    }

    \draw[azul,ultra thick] (e7) -- (g7) -- (s7); 

    \draw[lightgray] (1,0) grid (\nb,\nb);
    \draw[verde,ultra thick] (1,0) -- (\nb,0);
    \draw [rojo,very thick] (1,1) -- (\nb,\nb) (\nb,1) -- (1,\nb);

    \foreach \c in {\cem,...,\nb} 
    \draw[azul] (e\c) -- (g\c) -- (s\c);
    \node at (0,0) {$\sharp$}; \node at (\nb+1,0) {$\sharp$};
  \end{scope}
\begin{scope}[xshift=25cm]
   \def\nb{8}
     \def\demid{5}

     \fill[yellow!25] (1,0) -- (7,0) -- (1,6) -- cycle; 

    \foreach \c in {1,...,\nb}
    \coordinate[label=below:$x_{\c}$]  (e\c) at (\c,0);

    \foreach \c in {\demid,...,\nb} {
      \coordinate[label=above:$x_{\c}$] (s\c) at (\nb-\c+1,\c);
      \coordinate (g\c) at (\demid-\c/2-.5,3*\c/2-\demid+.5);
      \coordinate (d\c) at (\demid+\c/2-.5,\demid-\c/2-.5);
    }

    \draw[azul,ultra thick] (e7) -- (g7) -- (s7); 

    \draw[lightgray] (1,0) grid (\nb,\nb);
    \draw[verde,ultra thick] (1,0) -- (\nb,0);
    \draw [rojo,very thick] (1,1) -- (\nb,\nb) (\nb,1) -- (1,\nb);

    \foreach \c in {\demid,...,\nb} 
    \draw[azul] (e\c) -- (g\c) -- (s\c);
    \node at (0,0) {$\sharp$}; \node at (\nb+1,0) {$\sharp$};
  \end{scope}
  \begin{scope}[xshift=36cm]
   \def\nb{7}
   \def\demid{4.5}
   \def\cem{5}

   \fill[yellow!25] (1,0) -- (7,0) -- (1,6) -- cycle; 

    \foreach \c in {1,...,\nb}
    \coordinate[label=below:$x_{\c}$]  (e\c) at (\c,0);

    \foreach \c in {\cem,...,\nb} {
      \coordinate[label=above:$x_{\c}$] (s\c) at (\nb-\c+1,\c);
      \coordinate (g\c) at (\demid-\c/2-.5,3*\c/2-\demid+.5);
      \coordinate (d\c) at (\demid+\c/2-.5,\demid-\c/2-.5);
    }

    \draw[azul,ultra thick] (e7) -- (g7) -- (s7); 

    \draw[lightgray] (1,0) grid (\nb,\nb);
    \draw[verde,ultra thick] (1,0) -- (\nb,0);
    \draw [rojo,very thick] (1,1) -- (\nb,\nb) (\nb,1) -- (1,\nb);

    \foreach \c in {\cem,...,\nb} 
    \draw[azul] (e\c) -- (g\c) -- (s\c);
    \node at (0,0) {$\sharp$}; \node at (\nb+1,0) {$\sharp$};
  \end{scope}

\end{tikzpicture}
    
    \begin{tikzpicture}[scale=.5,font=\footnotesize,anchor=center]
  \def\nb{10}

  \fill[yellow!25] (1,0) -- (7,0) -- (1,6) -- cycle; 

    \foreach \c in {1,...,\nb}
    \coordinate[label=below:$x_{\c}$]  (e\c) at (\c,0);
    \node at (0,0) {$\sharp$}; \node at (\nb+1,0) {$\sharp$};
    \draw[help lines] (1,0) grid (\nb,\nb);
    \draw[verde,ultra thick] (1,0) -- (\nb,0);
    \draw [rojo,very thick] (1,1) -- (\nb,\nb) (\nb,1) -- (1,\nb);
    
    \foreach \i/\cem in {2/2,3/3,4/3,5/4,6/4,7/5,8/5,9/6,10/6} {
    \pgfmathparse{\i/2} \let\demid\pgfmathresult
    \foreach \c in {\cem,...,\i} {
      \coordinate (s\c) at (\i-\c+1,\c);
      \node[above,fill=white,inner sep=1pt] at (s\c) {$x_{\c}$}; 
      \coordinate (g\c) at (\demid-\c/2+.5,3*\c/2-\demid-.5);
    }
  }

  \foreach \i in {1,...,4}{
  \node[above,fill=white,inner sep=1pt] at (\i,7) {\bm{$x_7$}};
}
\end{tikzpicture}
    
    \caption{Triangles appearing in the space-time diagram of the
      mirror computation}
    \label{fig:trian-mirror}
\end{figure}
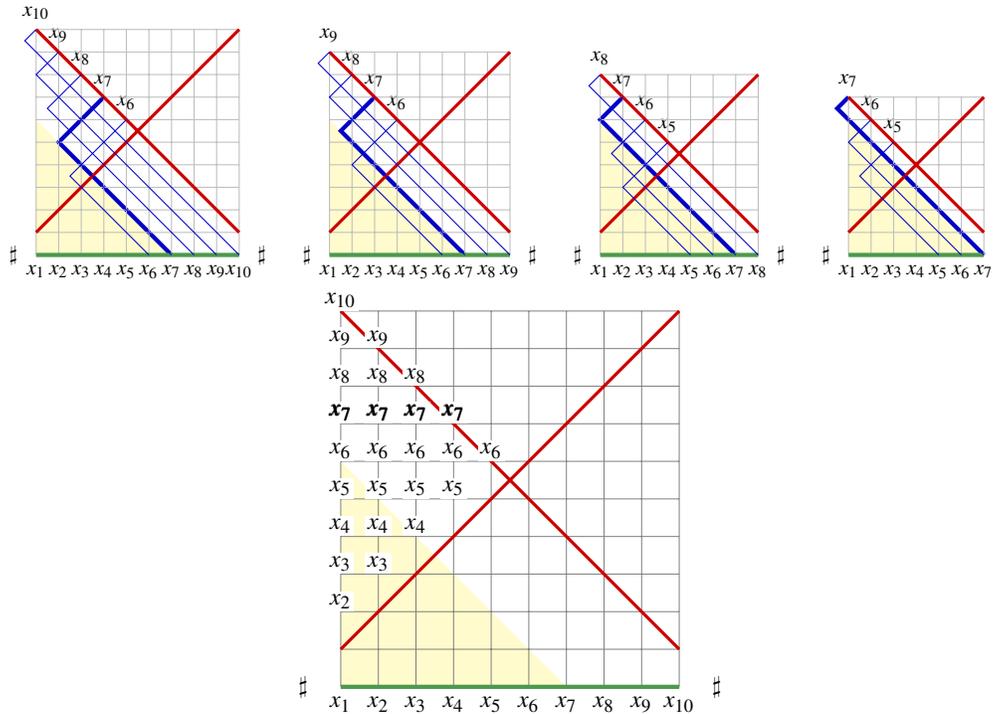

\end{proof}

This last result allows us to state a separation result between strict
and synchronized real time. Note that there exist other ways to
achieve such a result (using the separation of real time recognition
with central output or left output for example).

\begin{cor}
  Strict real time is strictly included in synchronous real time.
\end{cor}

As mirror can be achieved by composition of the two functions $f$ and
$g$ given in Example~\ref{ex:fg} which both are computable in strict real
time, it provides an explicit counter-example for the stability of this
class under composition.

\begin{cor}
  Strict real time is not closed under composition.
\end{cor}

\section{Conclusion} 
 
The parallel functional classes provide an  interesting perspective 
to study the question of small time complexities over cellular automata. 

The
three defined classes: strict real time, synchronous real time and linear time
form a chain of inclusions and all contain ``natural'' 
problems.
Not surprisingly, the functional classes are strongly linked with the recognition classes. In particular open questions about proper inclusions are correlated.

Moreover, the functional approach brings us new
solvable problems for which solutions are interesting outside
the strict range of functional classes (as the generic method used in
several algorithms or the result over restriction on uniform triangles
in space-time diagram).





\end{document}